\pdfoutput=1 

\newif\iftods
\todsfalse

\iftods
\documentclass[acmsmall]{acmart}
\setcopyright{cc}
\setcctype{by}
\acmJournal{TODS}
\acmYear{2025} \acmVolume{1} \acmNumber{1} \acmArticle{1} \acmMonth{1} \acmPrice{}\acmDOI{10.1145/3716375}
\else
\documentclass[acmsmall, nonacm]{acmart}
\fi

\iftods
\else
\fi

\usepackage[utf8]{inputenc}
\usepackage[english]{babel}
\usepackage{amsfonts}
\usepackage[showonlyrefs]{mathtools}
\usepackage{amsthm}
\usepackage{bbm}
\usepackage[shortlabels]{enumitem}
\usepackage[toc]{appendix}
\usepackage{microtype}
\usepackage[ruled,linesnumbered,noend]{algorithm2e}
\usepackage{xcolor}

\usepackage{algpseudocode}

\renewcommand{\epsilon}{\varepsilon}
\newcommand{\E}{\mathbb{E}}

\newcommand{\Universe}{\mathcal{U}}

\newcommand{\x}{{S}}

\newcommand{\X}{\mathcal{P}(\Universe)^{*}}
\newcommand{\Xsingle}{\Universe^{*}}

\newcommand{\floor}[1]{{\left \lfloor #1 \right \rfloor}}
\newcommand{\ceil}[1]{{\left \lceil #1 \right \rceil}}

\newcommand{\MG}{\mathrm{MG}}
\newcommand{\PMG}{\mathrm{PMG}}
\newcommand{\Laplace}{\mathrm{Laplace}}
\newcommand{\Lap}[1]{\Laplace{\left(#1\right)}}
\newcommand{\merge}{\mathrm{Merge}}
\newcommand{\usersketch}{\mathrm{PAMG}}
\newcommand{\gaussthreshold}{\mathrm{GSHM}}

\newcommand{\w}{w} 
\newcommand{\T}{T} 
\newcommand{\cc}{c} 
\newcommand{\TT}{\T \cap \T'} 

\newcommand{\numstreams}{l}

\newcommand{\thres}{2\ln(3/\delta)/\epsilon}

\SetKwFor{RepTimes}{repeat}{times}{end}

\newtheorem{theorem}{Theorem}
\newtheorem{corollary}[theorem]{Corollary}
\newtheorem{definition}[theorem]{Definition}
\newtheorem{lemma}[theorem]{Lemma}

\newtheorem{fact}[theorem]{Fact}

\theoremstyle{remark}

\makeatletter
\def\blfootnote{\xdef\@thefnmark{}\@footnotetext}
\makeatother

\title[Better Differentially Private Approximate Histograms and Heavy Hitters using the Misra-Gries Sketch]{\texorpdfstring{Better Differentially Private Approximate Histograms and Heavy Hitters using the Misra-Gries Sketch}{Better Differentially Private Approximate Histograms and Heavy Hitters using the Misra-Gries Sketch}}

\author{Christian Janos Lebeda}
\authornote{This work was carried out when Lebeda was at IT University of Copenhagen and Tětek was at University of Copenhagen. Both authors were at Basic Algorithms Research Copenhagen (BARC), supported by the VILLUM Foundation grant 16582.} 
\iftods
\email{christian-janos.lebeda@inria.fr}
\else
\email{christian.j.lebeda@gmail.com}
\fi
\affiliation{%
  \institution{Inria, University of Montpellier}
  \city{Montpellier}
  \country{France}
}

\author{Jakub Tětek}
\authornotemark[1]
\iftods
\email{jakub.tetek@insait.ai}
\else
\email{j.tetek@gmail.com}
\fi
\affiliation{%
  \institution{INSAIT, Sofia University “St. Kliment Ohridski”}
  \city{Sofia}
  \country{Bulgaria}
}

\iftods



\fi

\begin{document}

\begin{abstract}
    We consider the problem of computing differentially private approximate histograms and heavy hitters in a stream of elements. 
    In the non-private setting, this is often done using the sketch of Misra and Gries [Science of Computer Programming, 1982].
    Chan, Li, Shi, and Xu\ [PETS 2012] 
    describe a differentially private version of the Misra-Gries sketch, but the amount of noise it adds can be large and scales linearly with the size of the sketch; the more accurate the sketch is, the more noise this approach has to add. 
    We present a better mechanism for releasing a Misra-Gries sketch under $(\varepsilon,\delta)$-differential privacy. 
    It adds noise with magnitude independent of the size of the sketch; in fact, the maximum error coming from the noise is the same as the best known in the private non-streaming setting, up to a constant factor. 
    Our mechanism is simple and likely to be practical. 
    We also give a simple post-processing step of the Misra-Gries sketch that does not increase the worst-case error guarantee.
    It is sufficient to add noise to this new sketch with less than twice the magnitude of the non-streaming setting. 
    This improves on the previous result for $\epsilon$-differential privacy where the noise scales linearly to the size of the sketch.
    Finally, we consider a general setting where users can contribute multiple distinct elements. 
    We present a new sketch with maximum error matching the Misra-Gries sketch. 
    For many parameters in this setting our sketch can be released with less noise under $(\varepsilon, \delta)$-differential privacy.
\end{abstract}

\maketitle


\iftods
\else
\section*{Publication}
The conference version of this paper appears in Proceedings of the 42nd ACM SIGMOD-SIGACT-SIGAI Symposium on Principles of Database Systems, PODS 2023 \url{https://doi.org/10.1145/3584372.3588673}. 

A condensed version of the paper is available in the \href{https://sigmodrecord.org/2024/04/07/better-differentially-private-approximate-histograms-and-heavy-hitters-using-the-misra-gries-sketch/}{2024 SIGMOD Research Highlight Awards}.  

The full version is accepted for publication in the ACM Transactions on Database Systems \url{https://doi.org/10.1145/3716375}. 
\fi

\section{Introduction}
\label{sec:introduction}

Computing the histogram of a dataset is one of the most fundamental tasks in data analysis.
At the same time, releasing a histogram may present significant privacy issues. This makes the efficient computation of histograms under privacy constraints a fundamental algorithmic question. Notably, differential privacy has become in recent years the golden standard for privacy, giving formal mathematical privacy guarantees.
It would thus be desirable to have an efficient way of (approximately) computing histograms under differential privacy.

Histograms have been investigated thoroughly in the differentially private setting~\cite{dworkCalibratingNoise, DP-geometric-mechanism, DP-staircase-mechanism, DP-pure-sparse-hist, DP-approx-sparse-hist, DP-finite-computers, DP-ALP}. 
These algorithms start by computing the histogram exactly and they then add noise to ensure privacy. However, in practice, the amount of data is often so large that computing the histogram exactly would be impractical. 
This is, for example, the case when computing the histogram of high-volume streams such as when monitoring computer networks, online users, financial markets, and similar. In that case, we need an efficient streaming algorithm. 
Since the streaming algorithm would only compute the histogram approximately, the above-mentioned approach that first computes the exact histogram is infeasible. 
In practice, non-private approximate histograms are often computed using the Misra-Gries (MG) sketch \cite{Misra-Gries}. The MG sketch of size $k$ returns at most $k$ items and their approximate frequencies $\hat{f}$ such that $\hat{f}(x) \in [f(x) - n/(k+1), f(x)]$ for all elements $x$ where $f(x)$ is the true frequency and $n$ is the length of the stream. 
This error is known to be optimal~\cite{MG-max-error}. 
In this work, we develop a way of releasing a MG sketch in a differentially private way while adding only a small amount of noise. %
This allows us to efficiently and accurately compute approximate histograms in the streaming setting while not violating users' privacy. 
This can then be used to solve the heavy hitters problem in a differentially private way. 
Our result improves upon the work of Chan, Li, Shi, and Xu~\cite{DP-continual-heavy-hitters} who also show a way of privately releasing the MG sketch, but who need a greater amount of noise; we discuss this below.

In general, the issue with making approximation algorithms differentially private is that although we may be approximating a function with low global sensitivity, the algorithm itself (or rather the function it implements) may have a much larger global sensitivity. 
This increases the amount of noise required to achieve privacy using standard techniques.
We get around this issue by exploiting the structure of the difference between the MG sketches for neighboring inputs. 
This allows us to prove that the following simple mechanism 
ensures $(\epsilon,\delta)$-differential privacy: 
(1) We compute the Misra-Gries sketch, 
(2) we add to each counter independently noise distributed as $\Lap{1/\epsilon}$, 
(3) we add to all counters the same value, also distributed as $\Lap{1/\epsilon}$, 
(4) we remove all counters smaller than $1 + \thres$. 
Specifically, we show that this algorithm satisfies the following guarantees: 
\begin{theorem}[Theorem~\ref{thm:final_theorem} simplified]
    The above algorithm is $(\epsilon,\delta)$-differentially private, uses $2k$ words of space, and returns a frequency oracle $\hat{f}$ with maximum error of $n/(k+1) + O(\log (1/\delta)/\epsilon)$ with high probability for $\delta$ being sufficiently small.
\end{theorem}

A construction for a differentially private Misra-Gries sketch has been given before by Chan et al.~\cite{DP-continual-heavy-hitters}. 
However, the more accurate they want their sketch to be (and the bigger it is), their approach has to add \emph{more} noise. 
The reason is that they directly rely on the global $\ell_1$-sensitivity of the sketch.
Specifically, if the sketch has size $k$ (and thus error $n/(k + 1)$ on a stream of $n$ elements), its global sensitivity is $k$, and they thus have to add noise of magnitude $k/\epsilon$. 
Their mechanism ends up with an error of $O\left(k \log (d)/\epsilon\right)$ for $\epsilon$-differential privacy with $d$ being the universe size. 
This can be easily improved to $O\left(k \log (1/\delta)/\epsilon\right)$ for $(\epsilon, \delta)$-differential privacy with a thresholding technique similar to what we do in step $(4)$ of our algorithm above.
This also means that they cannot get more accurate than error $\Theta\left(\sqrt{n \log(1/\delta)/\epsilon}\right)$, no matter what value of $k$ one chooses.
We achieve that the biggest error, as compared to the values from the MG sketch, among all elements is $O(\log(1/\delta)/\epsilon)$ assuming $\delta$ is sufficiently small (we show more detailed bounds including the mean squared errors in Theorem~\ref{thm:final_theorem}). 
This is the same as the best private solution that starts with an exact histogram~\cite{DP-approx-sparse-hist}.
In fact, for any mechanism that outputs at most $k$ heavy hitters there exists input with error at least $n/(k+1)$ in the streaming setting~\cite{MG-max-error} and input with error at least $O(\log(\min(d, 1/\delta))/\epsilon)$ under differential privacy~\cite{DP-finite-computers}.
In Section~\ref{sec:pure} we discuss how to achieve $\epsilon$-differential privacy with error $n/(k+1) + O(\log(d)/\epsilon)$.
Therefore the error of our mechanisms is asymptotically optimal for approximate and pure differential privacy, respectively. 
The techniques used in Section~\ref{sec:pure} could also be used to get approximate differential privacy, but the resulting sketch would not have strong competitiveness guarantees with respect to the non-private Misra-Gries sketch, unlike the sketch that we give in Section~\ref{sec:alg}.

Chan et al.~\cite{DP-continual-heavy-hitters} use their differentially private Misra-Gries sketch as a subroutine for continual observation and combine sketches with an untrusted aggregator.
Those settings are not a focus of our work but our algorithm can replace theirs as the subroutine, leading to better results also for those settings.
However, the error from noise still increases linearly in the number of merges when the aggregator is untrusted.
As a side note, we show that in the case of a trusted aggregator, the approach of \cite{DP-continual-heavy-hitters} can handle merge operations without increasing error.
While that approach adds significantly more noise than ours if we do not merge, it can with this improvement perform better when the number of merges is very large (at least proportional to the sketch size). 

Furthermore, in Section~\ref{sec:user-level} we consider the setting where users contribute multiple elements to the stream. 
The algorithms discussed above can be adjusted for this setting by scaling the noise linearly proportional to the number of elements per user.
This is the optimal scaling for $\varepsilon$-differential privacy or when users can contribute many copies of the same element.
However, under $(\varepsilon, \delta)$-differential privacy in the non-streaming setting the magnitude of noise scales only with the square root of the number of elements if they are distinct~\cite{GaussianSparseHistogramMechanism}.
We give a generalized version of the Misra-Gries sketch with lower sensitivity for this setting.
We achieve less noise for many parameters when users contribute multiple distinct elements.
The properties of our new sketch is listed below.
It remains an open problem to achieve optimal error in this setting under $(\varepsilon, \delta)$-differential privacy.

\begin{theorem}[Lemmas~\ref{lem:l2sens-MG-sketch},~\ref{lem:user-level-sketch-error} and \ref{lem:user-level-sketch-sensitivity} summarized]
    The Privacy-Aware Misra Gries (PAMG) sketch and the MG sketch has the same worst-case error guarantees.
    If each user contributes a set of up to $m$ elements to the stream, then the sensitivity of the MG sketch scales linearly in $m$.
    As such, any differential private mechanism that outputs a MG sketch must add noise with magnitude $\Omega(m)$.
    In contrast, the $\ell_2$-sensitivity of the PAMG sketch is $\sqrt{k}$ independent of the value of $m$.
\end{theorem}

Another approach that can be used is to use a randomized frequency oracle to recover heavy hitters. 
However, it seems hard to do this with the optimal error size. 
In its most basic form \cite[Appendix~D]{DP-Google-shuffled-model}, this approach needs noise of magnitude $\Theta(\log(d)/\epsilon)$, even if we have a sketch with sensitivity 1 (the approach increases the sensitivity to $\log(d)$, necessitating the higher noise magnitude), leading to maximum error at least $\Omega(\log(k) \log(d)/\epsilon)$. 
Bassily, Nissim, Stemmer, and Guha Thakurta~\cite{bassily2017practical} show a more involved approach which reduces the maximum error coming from the noise to $\Theta((\log(k) + \log(d))/\epsilon)$, but at the cost of increasing the error coming from the sketch by a factor of $\log(d)$. 
This means that even if we had a sketch with error $\Theta(n/k)$ and sensitivity 1, neither of these two approaches would lead to optimal guarantees, unlike the algorithm we give in this paper.

\paragraph*{Relation to B{\"{o}}hler and Kerschbaum~\cite{multiparty-misra-gries}} 
Essentially the same result as Theorem~\ref{thm:final_theorem} has been claimed in~\cite{multiparty-misra-gries}. 
However, their approach ignores the discrepancy between the global sensitivity of a function we are approximating and that of the function the algorithm actually computes.
Their mechanism adds noise scaled to the sensitivity of the exact histogram which is $1$ when a user contributes a single element to the stream. 
But as shown by Chan et al.~\cite{DP-continual-heavy-hitters} the sensitivity of the Misra-Gries sketch scales linearly with the number of counters in the sketch. 
The algorithm from \cite{multiparty-misra-gries} thus does not achieve the claimed privacy parameters. 
Moreover, it seems unlikely this could be easily fixed -- not without doing something along the lines of what we do in this paper.
\\\\
See 
\href{https://github.com/JakubTetek/Differentially-Private-Misra-Gries}{https://github.com/JakubTetek/Differentially-Private-Misra-Gries} 
for sample implementations of some of the algorithms we present in this paper.

\section{Technical overview}
\paragraph*{Misra-Gries sketch} Since our approach depends on the properties of the MG sketch, we describe it here. 
Readers familiar with the MG sketch may wish to skip this paragraph. 
We describe the standard version; in Section~\ref{sec:alg} we use a slight modification, but we do not need that here. 

Suppose we receive a sequence of elements from some universe. 
At any time, we will be storing at most $k$ of these elements. 
Each stored item has an associated counter, other elements have implicitly their counter equal to $0$. 
When we process an element, we do one of the following three updates: 
(1) if the element is being stored, increment its counter by $1$, 
(2) if it is not being stored and the number of stored items is $<k$, store the element and set its counter to $1$, 
(3) otherwise decrement all $k$ counters by $1$ and remove those that reach $0$. 
The exact guarantees on the output will not be important now, and we will discuss them in Section~\ref{sec:alg}.

\paragraph*{Main contribution}
We now sketch how to release an MG sketch in a differentially private way. 

Consider two neighboring data streams $S = (S_1, \cdots, S_n)$ and $S' = (S_1, \cdots, S_{i-1}, S_{i+1}, \cdots, S_n)$ for some $i \in [n]$. 
At step $i-1$, the state of the MG sketch on both inputs is exactly the same. 
$MG_S$ then receives the item $S_i$ while $MG_{S'}$ does not. 
This either increments one of the counters of $MG_{S}$ (possibly by adding an element and raising its counter from $0$ to $1$) or decrements all its counters. 
In $\ell_1$ distance, the vector of the counters thus changes by at most $k$. 
One can show by induction that this will stay this way: at any point in time, $\|MG_S - MG_{S'}\|_1 \leq k$. 
By a standard global sensitivity argument, one can achieve pure DP by adding noise of magnitude $k/\epsilon$ to each count. 
This is the approach used in~\cite{DP-continual-heavy-hitters}. 
Similarly, we could achieve $(\epsilon, \delta)$-DP by using the Gaussian mechanism \cite{DworkRothBook} with noise magnitude proportional to the $\ell_2$-sensitivity, which is $\sup_{S,S'} \|MG_S - MG_{S'}\|_2 \leq \sqrt{k}$.
We want to instead achieve noise with magnitude $O(1/\epsilon)$ at each count. To this end, we need to exploit the structure of $MG_S - MG_{S'}$.

What we just described requires that we add the noise to the counts of all items in the universe, also to those that are not stored in the sketch. 
This results in the maximum error of all frequencies depending on the universe's size, which we do not want. 
However, it is known that this can be easily solved under $(\epsilon,\delta)$-differential privacy by only adding noise to the stored items and then removing values smaller than an appropriately chosen threshold~\cite{DP-approx-sparse-hist}. 
This may introduce additional error -- for this reason, we end up with error $O(\log(1/\delta)/\epsilon)$. 
As this is a somewhat standard technique, we ignore this in this section, we assume that the sketches $MG_S$ and $MG_{S'}$ store the same set of elements; the thresholding allows us to remove this assumption, while allowing us to add noise only to the stored items, at the cost of only getting approximate differential privacy.

We now focus on the structure of $MG_S - MG_{S'}$. 
After we add to $MG_S$ the element $S_i$, it either holds $(1)$ that $MG_S - MG_{S'}$ is a vector of all $0$'s and one $1$ or $(2)$ that $MG_S - MG_{S'} = -\mathbf{1}^k$ 
(We use $\mathbf{1}^k$ the denote the dimension $k$ vector of all ones).
We show by induction that this will remain the case as more updates are done to the sketches (note that the remainders of the streams are the same). 
We do not sketch the proof here, as it is quite technical.

How do we use the structure of $MG_S - MG_{S'}$ to our advantage? We add noise twice. 
First, we independently add to each counter noise distributed as $\Lap{1/\epsilon}$. Second, we add to all counters the same value, also distributed as $\Lap{1/\epsilon}$. 
That is, we release $MG_S + \Lap{1/\epsilon}^{\otimes k} + \Lap{1/\epsilon}\mathbf{1}^k$ 
(For $D$ being a distribution, we use $D^{\otimes k}$ to denote the $k$-dimensional distribution consisting of $k$ independent copies of $D$).
Intuitively speaking, the first noise hides the difference between $S$ and $S'$ in case $(1)$ and the second noise hides the difference in case $(2)$. 
We now sketch why this is so for worse constants: $2/\epsilon$ in place of $1/\epsilon$. 
When proving this formally, we use a more technical proof which leads to the better constant.

We now sketch why this is differentially private. 
Let $m_S$ be the mean of the counters in $MG_S$ for $S$ being an input stream. 
We may represent $MG_S$ as $(MG_S - m_S \mathbf{1}, m_S)$ (note that there is a bijection between this representation and the original sketch).
We now argue that the $\ell_1$-sensitivity of this representation is $< 2$ (treating the representation as a $(k+1)$-dimensional vector for the sake of computing the $\ell_1$ distances). 
Consider the first case. In that case, the averages $m_S,m_{S'}$ differ by $1/k$. 
As such, $MG_S - m_S \mathbf{1}^k$ and $MG_{S'} - m_{S'} \mathbf{1}^k$ differ by $1/k$ at $k-1$ coordinates and by $1-1/k$ at one coordinate. 
The overall $\ell_1$ change of the representation is thus 
\[
(k-1)\cdot \frac{1}{k} + (1-1/k)  + 1/k = 2-1/k < 2 \,.
\]
Consider now the second case when $MG_{S} - MG_{S'} = -\mathbf{1}^k$. 
Thus, $MG_S - m_S = MG_S' - m_{S'}$. 
At the same time $|m_S - m_{S'}|=1$. 
This means that the $\ell_1$ distance between the representations is $1$. 
Overall, the $\ell_1$-sensitivity of this representation is $< 2$.

This means that adding noise from $\Lap{2/\epsilon}^{\otimes k+1}$ to this representation of $MG_S$ satisfies $\epsilon$-DP. 
The resulting value after adding the noise is $(MG_S - m_S \mathbf{1}^k + \Lap{2/\epsilon}^{\otimes k}, m_S + \Lap{2/\epsilon})$. 
In the original vector representation of $MG_S$, this corresponds to $MG_S + \Lap{2/\epsilon}^{\otimes k} + \Lap{2/\epsilon}\mathbf{1}^k$ and, by post-processing, releasing this value is also differentially private. 
But this is exactly the value we wanted to show is differentially private! 

\paragraph*{Other settings}
Throughout this paper, we consider the problem of privately releasing a MG sketch in several settings.
As discussed above, we introduce a novel differentially private mechanism in Section~\ref{sec:alg} designed specifically for the sensitivity structure of the Misra-Gries sketch.
For the settings in Sections~\ref{sec:pure}, \ref{sec:merge}, and \ref{sec:user-level}, we instead use various other techniques that bounds the sensitivity.
We then apply general mechanisms from the differential privacy literature to release a sketch.
We briefly discuss our technical approach to each setting below.

In Section~\ref{sec:pure}, we consider the more restrictive definition of pure differential privacy ($\varepsilon$-DP).
The mechanism we discussed above does not satisfy $\varepsilon$-DP because it depends on the thresholding technique that hides the difference in stored keys between neighboring sketches with probability at least $1 - \delta$.
Instead, we introduce a simple post-processing procedure for a MG sketch.
We show that this post-processing adds additional error of at most $n/(k + 1)$ and reduces the $\ell_1$-sensitivity from $k$ to $2$.
This allows us to privately release the post-processed sketch by adding noise of magnitude independent of the sketch size.

In Section~\ref{sec:merge}, we consider the task of merging MG sketches from several streams, which is a common operation in practice.
There, we bound the sensitivity of the merging algorithm introduced by~\cite{mergeable-summaries} by showing that neighboring merged sketches differ by at most $1$ in each counter.

Finally, in Section~\ref{sec:user-level}, we consider a generalized setting where users contribute a set of up to $m$ elements to the stream.
The sensitivity of an MG sketch scales linearly with $m$ in this setting, which is undesirable under $(\varepsilon, \delta)$-DP, since it is sufficient to scale noise only with $\sqrt{m}$ in the non-streaming setting.
To address this challenge, we introduce a new sketch inspired by the MG sketch and designed with privacy in mind.
The key difference between our sketch and the standard MG sketch is that we decrement counters at most once per user rather than once per element.
We show that our sketch has the same worst-case error guarantees as the MG sketch and that neighboring sketches differ by at most $1$ in each counter, similar to the result in Section~\ref{sec:merge}.
In fact, our sketch can be seen as repeatedly running the merging algorithm of~\cite{mergeable-summaries}.
Our sketch can be released with less noise than the MG sketch for many parameters under $(\varepsilon, \delta)$-DP.
However, the error is still not optimal. It remains an open problem whether it is possible to design a sketch with noise of the same magnitude as the non-streaming setting and error guarantees matching the MG sketch.

\section{Preliminaries}
\label{sec:prelim-dpmg}

\paragraph*{Setup of this paper}
We use $\mathcal{U}$ to denote a universe of elements.  
We assume that $\mathcal{U}$ is a totally ordered set of size $d$.  
That is, $\mathcal{U}=[d]$ where $[d]=\{1,\ldots,d\}$.  
We write $\mathcal{P}(\Universe)$ to denote the power set of $\Universe$, and $\X$ to denote the set of all streams with finite length where each element of the stream is from $\mathcal{P}(\Universe)$.  
Given a stream $\x \in \X$ we want to estimate the frequency in $\x$ of each element of $\mathcal{U}$.   
We use $f(x) = \sum_{i = 1}^{\vert \x \vert} \mathbbm{1}[x \in \x_i]$ for any $x \in \mathcal{U}$ to denote the true frequency of $x$ in the stream $\x$, where $\mathbbm{1}[x \in \x_i]$ equals $1$ if $x \in \x_i$ and $0$ otherwise.  
In general, we have that each element in the stream is a set $\x_i \subseteq \mathcal{U}$, but all our results expect for Section~\ref{sec:user-level} are for the common setting where $\vert \x_i \vert = 1$.  
For simplicity of presentation, in this setting we denote the input as a stream of elements, that is $\Xsingle$, rather than a stream of sets each containing a single element.
Our algorithms output a set $T \subseteq \mathcal{U}$ of keys and a frequency estimate $c_i$ for all $i \in T$.  
The value $c_j$ is implicitly 0 for any $j \notin T$.
Our goal is to minimize the largest error between $c_x$ and $f(x)$ among all $x \in \Universe$.

\paragraph*{Differential privacy}
Differential privacy is a rigorous definition for describing the privacy loss of a randomized mechanism introduced by Dwork, McSherry, Nissim, and Smith~\cite{dworkCalibratingNoise}. 
Intuitively, differential privacy protects privacy by restricting how much the output distribution can change when replacing the input from one individual.
This is captured by the definition of neighboring datasets.  
We use the add-remove neighborhood definition for differential privacy.

\begin{definition}[Neighboring Streams]
\label{def:neighborhood}
    Let $\x$ be a stream of length $n$. 
    Streams $\x$ and $\x'$ are neighboring denoted $\x \sim \x'$ if there exists an $i \in [n + 1]$ such that $\x=(\x'_1,\ldots,\x'_{i-1},\x'_{i+1},\ldots,\x'_{n+1})$ or $\x'=(\x_1,\ldots,\x_{i-1},\x_{i+1},\ldots,\x_{n})$.
\end{definition}

\begin{definition}[Differential Privacy~\cite{DworkRothBook}]
\label{def:differential-privacy}
    A randomized mechanism $\mathcal{M}: \X \rightarrow \mathcal{R}$ satisfies $(\epsilon,\delta)$-differential privacy if and only if for all pairs of neighboring streams $\x \sim \x'$ and all measurable sets of outputs $Z \subseteq \mathcal{R}$ it holds that
    \[ \Pr[\mathcal{M}(\x) \in Z] \leq e^\epsilon \Pr[\mathcal{M}(\x') \in Z] + \delta \, . \]
\end{definition}


Samples from a Laplace distribution are used in many differentially private algorithms, most notably the Laplace mechanism~\cite{dworkCalibratingNoise}.
We write $\Lap{b}$ to denote a random variable with a Laplace distribution with scale $b$ centered around 0. 
We sometimes abuse notation and write $\Lap{b}$ to denote the value of a random variable drawn from the distribution. 
Our mechanism also works with other noise distributions. 
We briefly discuss this in Section~\ref{sec:practitioners}.

\begin{definition} [Laplace distribution] \label{def:lappdfcdf-dpmg}
  The probability density and cumulative distribution functions of the Laplace distribution centered around 0 with scale parameter $b$ are
  $f_b(x) = \frac{1}{2b}e^{-|x|/b}$, and
  $\Pr[\Lap{b} \leq x] = \frac{1}{2}e^{x/b}$ if $x < 0$ and $1 - \frac{1}{2}e^{-x/b}$ for $x \geq 0$.
\end{definition}

The sensitivity of a deterministic function taking a stream as input restricts the distance between the outputs for neighboring streams. 
We use the term sensitivity to describe any predicate that holds for the outputs for all pairs of neighboring streams. 
Two commonly used sensitivities in differential privacy are the $\ell_1$/$\ell_2$-sensitivities.
They are special cases of $\ell_p$-sensitivity which bounds the distance between outputs for neighboring streams in the $\ell_p$-norm.

\begin{definition} [$\ell_p$-sensitivity] \label{def:l-sensitivities}
  Let $g \colon \X \rightarrow \mathbb{R}^d$ be a deterministic function mapping a stream to a vector of real values.
  The $\ell_p$-sensitivity of $g$ for any $p \geq 1$ is 
  \[
  \Delta_p = \max_{\x \sim \x'} \| g(\x) - g(\x') \|_p \enspace ,
  \]
  where $\| g(\x) - g(\x') \|_p$ is the $\ell_p$-distance between $g(\x)$ and $g(\x')$ defined as
  \[
    \| g(\x) - g(\x') \|_p = \left({\sum_{i = 1}^d \vert g(\x) - g(\x') \vert^p } \right)^{1/p} \, .
  \]
\end{definition}

\section{Related work}
\label{sec:related}

Chan et al.~\cite{DP-continual-heavy-hitters} show that the global $\ell_1$-sensitivity of a Misra-Gries sketch is $\Delta_1=k$. 
(They actually show that the sensitivity is $k+1$ but they use a different definition of neighboring datasets that assumes $n$ is known. 
Applying their techniques under our definition yields sensitivity $k$.) 
They achieve privacy by adding 
noise with scale $k/\epsilon$ to all elements in the universe and keep the top-$k$ noisy counts. 
This gives an expected maximum error of $O(k\log(d)/\epsilon)$ with $\epsilon$-DP for $d$ being the universe size.
They use the algorithm as a subroutine for continual observation and merge sketches with an untrusted aggregator.
Those settings are not a focus of our work but our algorithm can replace theirs as the subroutine. 

Böhler and Kerschbaum~\cite{multiparty-misra-gries} worked on differentially private heavy hitters with no trusted server by using secure multi-party computation.
One of their algorithms adds noise to the counters of a Misra-Gries sketch. 
They avoid adding noise to all elements in the universe by removing noisy counts below a threshold which adds an error of $O(\log(1/\delta)/\epsilon)$.
This is a useful technique for hiding differences in keys between neighboring sketches that removes the dependency on $d$ in the error.
Unfortunately, as stated in the introduction their mechanism uses the wrong sensitivity.
The $\ell_1$-sensitivity of the sketch is $k$. 
If the magnitude of noise and the threshold are increased accordingly the error of their approach is $O(k \log(k/\delta) / \epsilon)$. 

If we ignore the memory restriction in the streaming setting, the problem is the same as the top-$k$ problem \cite{mir2011pan,durfee2019practical,carvalho2020differentially,qiao2021oneshot}.
The problem we solve can also be seen as a generalization of the sparse histogram problem. 
This has been investigated in~\cite{DP-pure-sparse-hist, DP-approx-sparse-hist, DP-finite-computers, DP-ALP}. 
Notably, Balcer and Vadhan~\cite{DP-finite-computers} provides a lower bound showing that for any $(\epsilon,\delta)$-differentially private mechanism that outputs at most $k$ counters, there exists input such that the expected error for some elements is $\Omega(\min(\log(d/k)/\epsilon, \log(1/\delta)/\epsilon,n))$ (assuming $\epsilon^2 > \delta$). 
The noise that we add in our main contribution in fact matches the second branch of the $\min$ over all elements.

A closely related problem is that of implementing frequency oracles in the streaming setting under differential privacy. 
This has been studied in e.g. \cite{DP-linear-sketches, DP-improved-count-sketch, DP-Google-shuffled-model}.
These approaches do not directly return the heavy hitters.
The simplest approach for finding the heavy hitters is to iterate over the universe which might be infeasible.
However, there are constructions for finding heavy hitters with frequency oracles more efficiently (see Bassily et al.~\cite{bassily2017practical}). 
However, as we discussed in the introduction, the approach of \cite{bassily2017practical} leads to worse maximum error than what we get unless the sketch size is very large and the universe size is small. 

The heavy hitters problem has also received a lot of attention in local differential privacy, starting with the paper introducing the RAPPOR mechanism \cite{erlingsson2014rappor} and continuing with \cite{qin2016heavy,zhao2022efficient,wang2019locally,bun2019heavy,wu2022asymptotically,bassily2017practical}. 
This problem is practically relevant, it is used for example by Apple to find commonly used emojis \cite{apple}. 
The problem has also been recently investigated when using cryptographic primitives \cite{zhu2020federated}.

\cite{blocki2022make,tvetek2022additive} have recently given general frameworks for designing differentially private approximation algorithms; however, if used naively, they are not very efficient for releasing multiple values (not more efficient than using composition) and they are thus not suitable for the heavy hitters problem.

In Section~\ref{sec:user-level} we consider a more general setting where each item in the stream is a set of size at most $m$ instead of a single element. 
Both \cite{DP-continual-heavy-hitters} and \cite{multiparty-misra-gries} also consider the setting where multiple elements differ between neighboring sketches and they achieve privacy by scaling the noise linearly to $m$. 
This scaling is typically used in mechanisms based on Laplace noise and is required for $\varepsilon$-DP even if we ignore the memory restriction of the streaming setting.
However, when the items are distinct we can instead add noise from the normal distribution. 
The magnitude of noise scales only with the square root of the number of differing counts. 
We use such a mechanism in Section~\ref{sec:user-level} based on work by Karrer, Kifer, Wilkins, and Zhang ~\cite{GaussianSparseHistogramMechanism}.
We discuss their results further in Theorem~\ref{thm:sparse-gaussian}.

\section{\texorpdfstring{Differentially Private Misra-Gries Sketch}{Differentially Private Misra-Gries Sketch}}
\label{sec:alg}

In this section, we present our algorithm for privately releasing Misra-Gries sketches.
We first present our variant of the non-private Misra-Gries sketch in Algorithm~\ref{alg:MG} and later show how we add noise to achieve $(\epsilon, \delta)$-differential privacy.
The algorithm we use differs slightly from most implementations of MG in that we do not remove elements that have weight $0$ until we need to re-use the counter. 
This will allow us to achieve privacy with slightly lower error. 

At all times, $k$ counters are stored as key-value pairs.
We initialize the sketch with dummy keys that are not part of $\Universe$.
This guarantees that we never output any elements that are not part of the stream, assuming we remove the dummy counters as post-processing.

The algorithm processes the elements of the stream one at a time. 
At each step one of three updates is performed:
(1) If the next element of the stream is already stored the counter is incremented by $1$. 
(2) If there is no counter for the element and all $k$ counters have a value of at least $1$ they are all decremented by $1$.
(3) Otherwise, one of the elements with a count of zero is replaced by the new element. 

In case (3) we always remove the smallest element with a count of zero. 
This allows us to limit the number of keys that differ between sketches for neighboring streams as shown in Lemma~\ref{lem:nb-cases}. 
The choice of removing the minimum element is arbitrary but the order of removal must be independent of the stream so that it is consistent between neighboring datasets.
The limit on differing keys allows us to obtain a slightly lower error for our private mechanism.
However, it is still possible to apply our mechanism with standard implementations of MG. 
We discuss this in Section~\ref{sec:MG-without-ordering}.

{
\begin{figure}[H]
\centering
\begin{minipage}{.85\linewidth}

\begin{algorithm}[H]
\caption{Misra-Gries~(MG) \label{alg:MG}}
\SetKwInOut{Input}{Input}

\Input{Positive integer $k$ and stream $\x \in \Xsingle$}

$T \leftarrow \{d + 1, \ldots, d + k\}$ \Comment{\makebox[4.85cm][l]{Start with $k$ dummy counters}} \\ 
$c_i \leftarrow 0$ for all $i \in T$ \\

\ForEach{$x \in \x$}
{
  \uIf(\Comment{\makebox[4.85cm][l]{Branch 1 - increment count}}){$x \in T$}{ 
    $c_x \leftarrow c_x + 1$ 
  }
  \uElseIf(\Comment{\makebox[4.85cm][l]{Branch 2 - decrement all counters}}){$c_i \geq 1$ for all $i \in T$}{
    $c_i \leftarrow c_i - 1$ for all $i \in T$ 
  }
  \Else(\Comment{\makebox[4.85cm][l]{Branch 3 - replace zero counter}}){ \label{branch3}
    Let $y \in T$ be the smallest key satisfying $c_y = 0$ \\
    $T \leftarrow \left( T \cup \{x\} \right) \setminus \{y\}$ \\
    $c_x \leftarrow 1$
  }
}

\Return{$T, c$}
\end{algorithm}

\end{minipage}
\end{figure}
}

The same guarantees about correctness hold for our version of the MG sketch, as for the original version. 
This can be easily shown, as the original version only differs in that it immediately removes any key whose counter is zero. 
Since the counters for items not in the sketch are implicitly zero, one can see by induction that the estimated frequencies by our version are exactly the same as those in the original version. 
We still need this modified version, as the set of keys it stores is different from the original version, which we use below. 
The fact that the returned estimates are the same however allows us to use the following fact

\begin{fact}[Bose, Kranakis, Morin, and Tang~\cite{MG-max-error}] \label{fact:MG-error}
Let $\hat{f}(x)$ be the frequency estimates given by an MG sketch of size $k$ for $n$ being the input size. 
Then it holds that $\hat{f}(x) \in [f(x) - n/(k+1), f(x)]$ for all $x \in \mathcal{U}$, where $f(x)$ is the true frequency of $x$ in $\x$.
\end{fact}

Note that this is optimal for any mechanism that returns a set of at most $k$ elements. 
This is easy to see for an input stream that contains $k + 1$ distinct elements each with frequency $n/(k+1)$ since at least one element must be removed.

We now analyze the value of $MG_S - MG_{S'}$ for $S, S'$ being neighboring inputs (recall Definition~\ref{def:neighborhood}). 
We will then use this in order to prove privacy.
As mentioned in Section~\ref{sec:related}, Chan et al.~\cite{DP-continual-heavy-hitters} showed that the $\ell_1$-sensitivity for Misra-Gries sketches is $k$.
They show that this holds after processing the element that differs for neighboring streams and use induction to show that it holds for the remaining stream.
Our analysis follows a similar structure.
We expand on their result by showing that the sets of stored elements for neighboring inputs differ by at most two elements when using our variant of Misra-Gries. 
We then show how all this can be used to get differential privacy with only a small amount of noise.

\begin{lemma} \label{lem:nb-cases}
    Let $\T,\cc \leftarrow \MG(k, \x)$ and $\T',\cc' \leftarrow \MG(k, \x')$ be the outputs of Algorithm~\ref{alg:MG} on a pair of neighboring streams $S \sim S'$ such that $S'$ is obtained by removing an element from $\x$.
    Then $\vert T \cap T' \vert \geq k - 2$ and all counters not in the intersection have a value of at most $1$.
    Moreover, it holds that either (1) $\cc_i = \cc'_i - 1$ for all $i \in T'$ and $\cc_j = 0$ for all $j \notin T'$ or (2) there exists an $i \in T$ such that $\cc_i = \cc'_i + 1$ and $\cc_j = \cc'_j$ for all $j \neq i$.    
\end{lemma}

\begin{proof}
    Let $\x \sim \x'$ be a pair of neighboring streams where $\x'$ is obtained by removing one element from $\x$. 
    We show inductively that the Lemma holds for any such $\x$ and $\x'$.
    Let $\w=\T-\T'$ and $\w'=\T'-\T$ denote the set of keys that are only in one sketch. 
    Let $\cc_0$ and $\cc'_0$ denote the smallest element with a zero count in the respective sketch when such an element exists.
    Then at any point during execution after processing the element removed from $\x$ the sketches are in one of the following states:
    
    \begin{enumerate}
        \item [(S1)] $\T = \T'$ and $\cc_i = \cc'_i - 1$ for all $i \in \T$.
        \item [(S2)] There exist $x_1, x_2 \in \Universe$ such that $\w=\{ x_1 \}$ and $\w' = \{x_2\}$, $\cc_{x_1} = 0$, $\cc'_{x_2} = 1$ and $\cc_i = \cc'_i - 1$ for all $i \in \TT$.
        \item [(S3)] $\T = \T'$ and there exists $x_1 \in \Universe$ such that $\cc_{x_1} = \cc'_{x_1} + 1$ and $\cc_{i} = \cc'_{i}$ for all $i \in \T \setminus \{ x_1 \}$.
        \item [(S4)] There exist $x_1, x_2 \in \Universe$ such that $\w=\{ x_1 \}$ and $\w' = \{x_2\}$, $\cc_{x_1} = 1$, $\cc'_{x_2} = 0$ and $\cc_{i}=\cc'_i$ for all $i \in \TT$.
        \item [(S5)] There exist $x_1, x_2, x_3 \in \Universe$ such that $\cc_{x_1} = \cc'_{x_1} + 1$, $\w=\{x_2\}$, $\w' = \{x_3\}$, $\cc_{x_2}=0$, $\cc'_{x_3}=0$ and $\cc_{i}=\cc'_i$ for all $i \in \TT \setminus \{ x_1 \}$.
        \item [(S6)] There exist $x_1, x_2, x_3, x_4 \in \Universe$ such that $\w=\{ x_1, x_2 \}$ and $\w' = \{x_3, x_4\}$, $\cc_{x_1} = 1$, $\cc_{x_2}=\cc'_{x_3}=\cc'_{x_4}=0$, $\cc_{i}=\cc'_i$ for all $i \in \TT$ and $x_4=\cc'_0$.
    \end{enumerate}

    Let $x=\x_i$ be the element in stream $\x$ which is not in stream $\x'$.
    Since the streams are identical in the first $i - 1$ steps the sketches are clearly the same before step $i$.
    If there is a counter for $x$ in the sketch we execute Branch 1 and the result corresponds to state S3.
    If there is no counter for $x$ and no zero counters we execute Branch 2 and the result corresponds to state S1.
    Otherwise, the 3rd branch of the algorithm is executed and $\cc_0$ is replaced by $x$ which corresponds to state S4.
    Therefore we must be in one of the states S1, S3, or S4 for $T, c \leftarrow \MG(k, (\x_1, \ldots, \x_i))$ and $T', c' \leftarrow \MG(k, (\x'_1, \ldots, \x'_{i-1}))$.

    We can then prove inductively that the Lemma holds since the streams are identical for the elements $(\x_{i+1},\ldots, \x_n)$.
    We have to consider the possibility of each of the branches being executed for both sketches.
    The simplest case is when the element has a counter in both sketches and Branch 1 is executed on both inputs. 
    This might happen in all states and we stay in the same state after processing the element. 
    But many other cases lead to new states.
    
    Below we systematically consider all outcomes of processing an element $x \in \Universe$ when the sketches start in each  of the above states.
    When processing each element, one of the three branches is executed for each sketch.
    This gives us up to 9 combinations to check, although some are impossible for certain states.
    Furthermore, when Branch 3 is executed we often have to consider which element is replaced as it leads to different states.
    We refer to $T,c$ and $T',c'$ as sketch 1 and sketch 2, respectively.
    
    
    {\bf State S1:} 
    If $x \in \T$ then $x \in \T'$ and Branch 1 is executed for both sketches. 
    Similarly, if Branch 2 is executed for sketch 1 it must also be executed for sketch 2 as all counters are strictly larger.
    Therefore we stay in state S1 in both cases.
    It is impossible to execute Branch 3 for sketch 2 since all counters are non-zero by definition.
    As such the final case to consider is when $x \notin \T$ and there is a counter with value $0$ in sketch 1.
    In this case, we execute Branch 3 for sketch 1 and Branch 2 for sketch 2.
    This results in state S4.

    {\bf State S2:}
    If $x \in \T$ we execute Branch 1 on sketch 1 and there are two possible outcomes.
    If $x \neq x_1$ we also execute Branch 1 on sketch 2 and remain in state S2.
    If $x = x_1$ we execute Branch 2 on sketch 2 in which case there are no changes to $T$ or $T'$ but now $c_x = 1$ and $c_i = c'_i$ for all $i \in \TT$. 
    As such, we transitioned to state S4. 
    
    Since $c_{x_1}=0$ by definition Branch 2 is never executed on sketch 1 and Branch 3 is never executed on sketch 2 as all counters are non-zero. 
    If $x=x_2$ Branch 3 is executed on sketch 1 and Branch 1 is executed for sketch 2. 
    If $c_0 = x_1$ the sketches have the same keys after processing $x$ and transition to state S1, otherwise if $c_0 \neq x_1$ the sketches still differ for one key and remain in state S2.
    
    Finally, if Branch 3 is executed on sketch 1 and Branch 2 is executed on sketch 2 we again have two possibilities.
    In both cases, the sketches store the same count on all elements from $\TT$ after processing $x$.
    If $c_0 = x_1$ it is removed from $T$ and replaced by $x$ with $c_x=1$ which corresponds to state S4.
    If $c_0 \neq x_1$ we must have that $c_0 \in \TT$.
    The two sketches now differ on exactly two keys after processing $x$.
    One of the two keys stored in sketch 2 that are not in sketch 1 must be the minimum zero key since the elements $c_0$ and $x_2$ now have counts of zero in sketch 2 and $c_0$ was the minimum zero key in $\TT$. 
    Therefore we transition to state S6.

    {\bf State S3:}
    The simplest case is $x \in \T$ since then $x \in \T'$ and Branch 1 is executed for both sketches. 
    If Branch 2 is executed for sketch 1 and $c'_{x_1} \neq 0$ Branch 2 is also executed for sketch 2.
    For both cases, we remain in state S3.
    Instead, if $c'_{x_1} = 0$ Branch 3 is executed for sketch 2.
    Since all counters are decremented for sketch 1 and $x_1$ is replaced in sketch 2 we transition to state S2.
    Lastly, if Branch 3 is executed for sketch 1 it is also executed for sketch 2 and there are two cases.
    If the same element is removed we remain in state S3.
    Otherwise, if $x_1$ is replaced in sketch 2 we transition to state S4.

    {\bf State S4:}
    Sketch 2 contains a counter with a value of zero in states S4, S5, and S6. Therefore Brach 2 is never executed on sketch 2 in the rest of the proof.
    If Branch 1 is executed for both sketches we stay in the same state as always but if $x=x_1$ Branch 1 is executed for sketch 1 and Branch 3 is executed for sketch 2.
    If $c'_0=x_2$ then $T=T'$ after processing $x$ and we transition to state S3.
    If $c'_0 \neq x_2$ another element is removed from sketch 2 which must also have a count of zero in sketch 1 and we go to state S5. 
    
    If Branch 2 is executed on sketch 1 we know that $c'_{x_2}$ must be the only zero counter in sketch 2.
    Therefore it does not matter if Branch 1 or 3 is executed on sketch 2.
    For both cases, we set $c_x = 1$ and the sketches differ in one key which corresponds to state S2.
    
    Finally, if Branch 3 is executed on sketch 1 we again have two cases that lead to the same state.
    If $x=x_2$ or $c'_0 = x_2$ the counter $c'_{x_2}$ is updated or replaced but the counter that was removed from sketch 1 still remains in sketch 2.
    Otherwise, we have $c_0=c'_0$ and we replace the same counter in sketches 1 and 2.
    Therefore we remain in state S4 in both cases. 

    {\bf State S5:}
    Since by definition both sketches contain counters with a value of zero, Branch 2 is never executed while in this state.
    If $x \in \TT$ we remain in the same state as always. 
    We have to consider the cases where $x = x_2$, $x = x_3$, and $x \notin T \cup T'$.
    The resulting state depends on the elements that are replaced in the sketch.
    For $x = x_2$ we transition to state S3 if $c'_0 = x_3$ and remain in state S5 otherwise.
    The same argument shows that we transition to state S3 or S5 based on $c_0$ if $x=x_3$.
    When $x \notin T \cup T'$ we execute Branch 3 on both sketches. 
    We transition to state S3 only if $c_0 = x_2$ and $c'_0 = x_3$ since otherwise both sketches still have a zero counter that is not stored in the other sketch and we stay in state S5.

    {\bf State S6:}
    Similar to state S5, Branch 2 is never executed from this state.
    Here we have to consider the five cases where $x \in \TT$, $x = x_1$, $x=x_2$, $x \in w'$, and $x \notin T \cup T'$.
    We know that $x_4$ is replaced whenever $x \notin T'$.
    If $x \in \TT$ we execute Branch 1 on both sketches and remain in state S6.
    If $x = x_1$ we transition to state S5 and for $x = x_2$ we transition to state S4.
    When $x \in w'$ there are two possibilities. 
    We always have $c_x = c'_x$ after updating. 
    If $c_0 = x_2$ the sketches will share $k-1$ keys and transition to state S4.
    If $c_0 \neq x_2$ then another element that has a count of zero in both sketches is replaced in sketch 1.
    We know that either this element or the remaining zero-valued element of $w'$ must be the smallest zero-valued element in sketch 2.
    Therefore we remain in state S6.
    
    The final case to consider is when $x \notin T \cup T'$. 
    In this case Branch, 3 is executed for sketch 2 and $x_4$ is replaced with $x$ in $T'$.
    If $c_0 = x_2$ we transition to state S4.
    Otherwise, either $x_3$ or the element that was replaced from sketch 1 must be the minimum element with a count of zero in sketch 2.
    As such, we remain in state S6.
%
   %
%
\end{proof}

Next, we consider how to add noise to release the Misra-Gries sketch under differential privacy.
Recall that Chan et al.~\cite{DP-continual-heavy-hitters} achieves privacy by adding noise to each counter which scales with $k$.
We avoid this by utilizing the structure of sketches for neighboring streams shown in Lemma~\ref{lem:nb-cases}.
We sample noise from $\Lap{1/\epsilon}$ independently for each counter, but we also sample one more random variable from the same distribution which is added to all counters.
Small values are then discarded using a threshold to hide differences in the sets of stored keys between neighboring inputs. 
This is similar to the technique used by e.g.~\cite{DP-approx-sparse-hist}.
The algorithm takes the output from $\MG$ as input. 
We sometimes write $\PMG(k, \x)$ as a shorthand for $\PMG(\MG(k, \x))$.
{
\begin{figure}[H]
\centering
\begin{minipage}{.85\linewidth}

\begin{algorithm}[H]
\caption{Private Misra-Gries~(PMG) \label{alg:PMG}}
\SetKwInOut{Input}{Input}
\SetKwInOut{Parameters}{Parameters}

\Parameters{$\epsilon, \delta > 0$}
\Input{Output from Algorithm~\ref{alg:MG}: $T, c \leftarrow \text{MG}(k, \x)$}

$\tilde{T} \leftarrow \emptyset$ \\
Sample $\eta \sim \Lap{1/\epsilon}$ \\

\ForEach{$x \in T$}
{
    $c_x \leftarrow c_x + \eta + \Lap{1/\epsilon}$ \Comment{\makebox[5.15cm][l]{Add the two noise samples to count}} \\
    \If(\Comment{\makebox[5.15cm][l]{Small noisy values are removed}}){$c_x \geq 1 + \thres$} {\label{line:weird_condition}
        $\tilde{T} \leftarrow \tilde{T} \cup \{x\}$ \\
        $\tilde{c}_x \leftarrow c_x$
    }
}
\Return{$\tilde{T}, \tilde{c}$}
\end{algorithm}

\end{minipage}
\end{figure}
}

We prove the privacy guarantees in three steps. 
First, we show that changing either a single counter or all counters by 1 does not change the output distribution significantly (Corollary~\ref{cor:eps-close}). 
This assumes that, for neighboring inputs, the set of stored elements is exactly the same.
By Lemma~\ref{lem:nb-cases}, we have that the difference between the sets of stored keys is small and the corresponding counters are $\leq 1$. 
Relying on the thresholding, we bound the probability of outputting one of these keys (Lemma~\ref{lem:nonsupport-delta}).
Finally, we combine these two lemmas to show that the privacy guarantees hold for all cases (we do this in Lemma~\ref{lem:PMG-privacy}).

\begin{lemma} \label{lem:vectors_inequality}
    Let us have $x, x' \in \mathbb{R}^k$ such that one of the following three cases holds
    \begin{enumerate}
        \item $\exists i \in [k]$ such that $| x_i - x'_i | = 1$ and $x_j = x'_j$ for all $j \neq i$.
        \item $x_i = x'_i - 1$ for all $i \in [k]$.
        \item $x_i = x'_i + 1$ for all $i \in [k]$.
    \end{enumerate}
    Then we have for any measurable set $Z$ that
 \[
 \Pr[x + \Laplace(1/\epsilon)^{\otimes k} + \Lap{1/\epsilon} \mathbf{1}^k\in Z] 
 \leq e^{\epsilon} \Pr[x' + \Lap{1/\epsilon}^{\otimes k} + \Lap{1/\epsilon} \mathbf{1}^k\in Z] .
 \]
\end{lemma}
\begin{proof}
Throughout the proof, we construct sets by applying a translation to all elements of another set. 
That is, for any 
$\phi \in \mathbb{R}^k$ and measurable set $Z$ we define $Z - \phi = \{a \in \mathbb{R}^k \vert a + \phi \in Z\}$.
We first focus on the simpler case $(1)$. It holds by the law of total expectation that
\begin{align*}
\Pr[x + \Lap{1/\epsilon}^{\otimes k} + \Lap{1/\epsilon} \mathbf{1}^k \in Z] &=\\ 
\E_{N \sim \Lap{1/\epsilon}}\big[\Pr[\Lap{1/\epsilon}^{\otimes k} \in Z - x - N\mathbf{1}^k|N]\big] &\leq\\ 
e^\epsilon \E_{N \sim \Lap{1/\epsilon}}\big[\Pr[\Lap{1/\epsilon}^{\otimes k} \in Z - x' - N\mathbf{1}^k|N]\big] &=\\ 
e^\epsilon \Pr[x' + \Lap{1/\epsilon}^{\otimes k} + \Lap{1/\epsilon}\mathbf{1}^k \in Z]& \,,
\end{align*}
where to prove the inequality, we used that for any measurable set $A$, it holds
\[
\Pr[\Lap{1/\epsilon}^{\otimes k} \in A] \leq e^\epsilon \Pr[\Lap{1/\epsilon}^{\otimes k} \in A - \phi]
\]for any $\phi \in \mathbb{R}^k$ with $\|\phi\|_1 \leq 1$ (see \cite{dworkCalibratingNoise}).
Specifically, we have set $A = Z - x - N\mathbf{1}^k$ and $\phi = x - x'$ such that $\|\phi\|_1 = 1$.

We now focus on the cases $(2)$ and $(3)$. 
We will prove below that for $x,x'$ satisfying one of the conditions $(2),(3)$ and for any measurable $A$, $Z$ and $N_1 \sim \Lap{1/\epsilon}^{\otimes k}$, it holds
\begin{align*} \label{eq:to_be_proved}
\Pr[x + N_1 + \Lap{1/\epsilon} \mathbf{1}^k \in Z | N_1 \in A] 
\leq e^\epsilon \Pr[x' + N_1 + \Lap{1/\epsilon} \mathbf{1}^k \in Z | N_1 \in A] \,.
\end{align*}
This allows us to argue like above:
\begin{align*}
\Pr[x + \Lap{1/\epsilon}^{\otimes k} + \Lap{1/\epsilon} \mathbf{1}^k \in Z] &=\\ 
\E_{N_1 \sim \Lap{1/\epsilon}^{\otimes k}}\big[\Pr[x + N_1 + \Lap{1/\epsilon} \mathbf{1}^k \in Z | N_1 ]\big] &\leq\\ 
e^\epsilon \E_{N_1 \sim \Lap{1/\epsilon}^{\otimes k}}\big[\Pr[x' + N_1 + \Lap{1/\epsilon} \mathbf{1}^k \in Z | N_1 ]\big] &=\\ 
e^\epsilon \Pr[x' + \Lap{1/\epsilon}^{\otimes k} + \Lap{1/\epsilon} \mathbf{1}^k \in Z] & \,,
\end{align*}
which would conclude the proof. 
Let $g: \mathbb{R} \rightarrow \mathbb{R}^k$ be the function $g(a) = a \mathbf{1}^k$ and define $g^{-1}(B) = \{a \in \mathbb{R} \vert g(a) \in B\}$ and note that $g$ is measurable. 
We focus on the case $(2)$; the same argument works for $(3)$ as we discuss below. It holds
\begin{align*}
\Pr[x + N_1 + \Lap{1/\epsilon} \mathbf{1}^k \in Z | N_1 \in A] &=\\ 
\Pr[\Lap{1/\epsilon} \mathbf{1}^k \in Z - x - N_1| N_1 \in A] &=\\ 
\Pr[\Lap{1/\epsilon} \in g^{-1}(Z - x - N_1)| N_1 \in A] &=\\ 
\Pr[\Lap{1/\epsilon} \in g^{-1}(Z - x' - \mathbf{1}^k - N_1)| N_1 \in A] &=\\ 
\Pr[\Lap{1/\epsilon} \in g^{-1}(Z - x' - N_1) - 1| N_1 \in A] &\leq\\ 
e^\epsilon \Pr[\Lap{1/\epsilon} \in g^{-1}(Z - x' - N_1)| N_1 \in A] &=\\ 
e^\epsilon \Pr[\Lap{1/\epsilon} \mathbf{1}^k \in Z - x' - N_1| N_1 \in A] &=\\ 
e^\epsilon \Pr[x' + N_1 + \Lap{1/\epsilon} \mathbf{1}^k \in Z | N_1 \in A] & \,.
\end{align*}
To prove the inequality, we again used the standard result that for any measurable $A$, it holds that $\Pr[\Lap{1/\epsilon} \in A] \leq e^{\epsilon} \Pr[\Lap{1/\epsilon} \in A-1]$. 
The same holds for $A+1$; this allows us to use the exact same argument in case $(3)$, in which the proof is the same except that $-1$ on lines 4,5 of the manipulations is replaced by $+1$.
%
%
\end{proof}

\begin{corollary} \label{cor:eps-close}
    Let $T,c$ and $T',c'$ be two sketches such that $T=T'$ and one of following holds:
    \begin{enumerate}
        \item $\exists i \in \T$ such that $| c_i - c'_i | = 1$ and $c_j = c'_j$ for all $j \neq i$.
        \item $c_i = c'_i - 1$ for all $i \in T$.
        \item $c_i = c'_i + 1$ for all $i \in T$.
    \end{enumerate}
    Then for any measurable set of outputs $Z$, we have:
    \[\Pr[\PMG(T,c) \in Z] \leq e^\epsilon \Pr[\PMG(T',c') \in Z]\,.\]
\end{corollary}
\begin{proof}
Consider first a modified algorithm $\PMG'$ that does not perform the thresholding: that is, if we remove the condition on line~\ref{line:weird_condition}. 
It can be easily seen that $\PMG'$ only takes the vector $c$ and releases $c + \Lap{1/\epsilon}^{\otimes k} + \Lap{1/\epsilon} \mathbf{1}^k$. 
We have just shown in Lemma~\ref{lem:vectors_inequality} that this means that for any measurable $Z'$,
$$\Pr[\PMG'(T,c) \in Z'] \leq e^\epsilon \Pr[\PMG'(T',c') \in Z']\,.$$

Let $\tau(x) = x$ for $x \geq 1+2\ln(3/\delta)/\epsilon$ and $0$ otherwise. 
Since $\PMG(T,c) = \tau(\PMG'(T,c))$, it then holds 
\begin{align*} 
\Pr[\PMG(T,c) \in Z] = & \Pr[\PMG'(T,c) \in \tau^{-1}(Z)] \leq\\  & e^\epsilon \Pr[\PMG'(T',c') \in \tau^{-1}(Z)] = e^\epsilon \Pr[\PMG(T',c') \in Z] \,,
\end{align*}
as we wanted to show.
\end{proof}

Next, we bound the effect on the output distribution from keys that differ between sketches by $\delta$.

\begin{lemma} \label{lem:nonsupport-delta}
    Let $T,c$ and $T',c'$ be two sketches of size $k$ and let $\hat{T} = T \cap T'$. 
    If we have that $\vert \hat{T} \vert \geq k - 2$, $c_i=c'_i$ for all $i \in \hat{T}$,
    and for all $x \notin \hat{T}$, it holds $c_x,c'_x \leq 1$. 
    Then for any measurable set $Z$, it holds
    $$\Pr[\PMG(T,c) \in Z] \leq \Pr[\PMG(T',c') \in Z] + \delta \,.$$ 
\end{lemma}

\begin{proof}
    Let $\PMG'(T,c)$ denote a mechanism that executes $\PMG(T, c)$ and post-processes the output by discarding any elements not in $\hat{T}$.
    It is easy to see that $(a) \,\,\, \Pr[\PMG'(T,c) \in Z] = \Pr[\PMG'(T',c') \in Z]$ since the input sketches are identical for all elements in $\hat{T}$.
    Moreover, for any output $\tilde{T}, \tilde{c} \leftarrow \PMG(T,c)$ for which $\tilde{T} \subseteq \hat{T}$, the post-processing does not affect the output. 
    This gives us the following inequalities: 
    $(b) \,\,\, \Pr[\PMG(T,c) \in Z] \leq \Pr[\PMG'(T,c) \in Z] + \Pr[\tilde{T} \not\subseteq \hat{T}]$ and 
    $(c) \,\,\, \Pr[\PMG'(T',c') \in Z] \leq \Pr[\PMG(T,c) \in Z] + \Pr[\tilde{T}' \not\subseteq \hat{T}]$. 
    Combining equations $(a) - (c)$, we get the inequality $\Pr[\PMG(T,c) \in Z] \leq \Pr[\PMG(T',c') \in Z] + \Pr[\tilde{T} \not\subseteq \hat{T}] + \Pr[\tilde{T}' \not\subseteq \hat{T}]$. 

    As such, the Lemma holds if $\Pr[\tilde{T} \not\subseteq \hat{T}] + \Pr[\tilde{T}' \not\subseteq \hat{T}] \leq \delta$. 
    That is, it suffices to prove that with probability at most $\delta$ any noisy count for elements not in $\hat{T}$ is at least $1 + \thres$.
    The noisy count for such a key can only exceed the threshold if one of the two noise samples added to the key is at least $\ln(3/\delta)/\epsilon$.
    From Definition~\ref{def:lappdfcdf-dpmg} we have $\Pr[\Lap{1/\epsilon} \geq \ln(3/\delta)/\epsilon] = \delta/6$.
    There are at most 4 keys not in $\hat{T}$ which are in $T \cup T'$ and therefore at most $6$ noise samples affect the probability of outputting such a key (the 4 individual Laplace noise samples and the 2 global Laplace noise samples, one for each sketch).
    By a union bound the probability that any of these samples exceeds $\ln(3/\delta)/\epsilon$ is at most $\delta$.
    %
\end{proof}

We are now ready to prove the privacy guarantee of Algorithm~\ref{alg:PMG}.

\begin{lemma} \label{lem:PMG-privacy}
Algorithm~\ref{alg:PMG} is $(\epsilon, \delta)$-differentially private for any $k$.
\end{lemma}

\begin{proof}
    The Lemma holds if and only if for any pair of neighboring streams $\x \sim \x'$ and any measurable set $Z$ we have:
    $$ \Pr[\PMG(T,c) \in Z] \leq e^\epsilon \Pr[\PMG(T',c') \in Z] + \delta \,,$$
    where $T, c \leftarrow \MG(k, \x)$ and $T', c' \leftarrow \MG(k, \x')$ denotes the non-private sketches for each stream.

    We prove the guarantee above using an intermediate sketch that ``lies between'' $T,c$ and $T',c'$.
    The sketch has support $T'$ and we denote the counters as $\hat{c}$. By Lemma~\ref{lem:nb-cases}, we know that $\vert T \cap T' \vert \geq k - 2$ and all counters in $c$ and not in $T \cap T'$ are at most 1. 
    We will now present some conditions on $\hat{c}$ such that if these conditions hold, the lemma follows. 
    We will then prove the existence of such $\hat{c}$ below.
    Assume that $\hat{c}_i=c_i$ for all $i \in T \cap T'$ and $\hat{c}_j \leq 1$ for all $j \in T' \setminus T$.  
    Lemma~\ref{lem:nonsupport-delta} then tells us that
    \[ 
        \Pr[\PMG(T,c) \in Z] \leq \Pr[\PMG(T',\hat{c}) \in Z] + \delta \,.
    \]
    Assume also that one of the cases required for Corollary~\ref{cor:eps-close} holds between $\hat{c}$ and $c'$. We have
    \[ 
        \Pr[\PMG(T',\hat{c}) \in Z] \leq e^\epsilon \Pr[\PMG(T',c') \in Z] \,. 
    \]
    Therefore, if such a sketch $T',\hat{c}$ exists for all $\x$ and $\x'$ the lemma holds since
    \[
        \Pr[\PMG(T,c) \in Z] \leq \Pr[\PMG(T',\hat{c}) \in Z] + \delta 
        \leq e^\epsilon \Pr[\PMG(T',c') \in Z] + \delta \, . 
    \]
    

    It remains to prove the existence of $\hat{c}$ such that $\hat{c}_i=c_i$ for all $i \in T \cap T'$ and $\hat{c}_j \leq 1$ for all $j \in T' \setminus T$ and such that one of the conditions $(1)-(3)$ of Corollary~\ref{cor:eps-close} holds between $\hat{c}$ and $c'$.
    We first consider neighboring streams where $\x'$ is obtained by removing an element from $\x$. 
    From Lemma~\ref{lem:nb-cases} we have two cases to consider.
    If $c_i = c'_i - 1$ for all $i \in T'$ we simply set $\hat{c} = c$. 
    Recall that we implicitly have $c_i = 0$ for $i \notin T$.
    Therefore the sketch satisfies the two conditions above since $\hat{c}_i = c_i$ for all $i \in \Universe$ and condition (2) of Corollary~\ref{cor:eps-close} holds. 
    In the other case where $c_i = c'_i + 1$ for exactly one $i \in T$ there are two possibilities.
    If $i \in T'$ we again set $\hat{c} = c$. 
    When $i \notin T'$ there must exist at least one element $j \in T'$ such that $c'_j = 0$ and $j \notin T$.
    We set $\hat{c}_j = 1$ and $\hat{c}_i = c'_i$ for all $i \neq j$. 
    In both cases $\hat{c}_i = c_i$ for all $i \in T \cap T'$ and $\hat{c}_j$ is at most one for $j \notin T$. 
    There is exactly one element with a higher count in $\hat{c}$ than $c'$ which means that condition (1) of Corollary~\ref{cor:eps-close} holds. 

    If $\x$ is obtained by removing an element from $\x'$ the cases from Lemma~\ref{lem:nb-cases} are flipped.
    If $c_i - 1 = c'_i$ for all $i \in T$ and $c'_j = 0$ for $j \notin T$ we set $\hat{c}_i = c_i$ if $i \in T$ and $\hat{c}_i = 1$ otherwise.
%
    It clearly holds that $\hat{c}_i=c_i$ for all $i \in T \cap T'$ and $\hat{c}_j \leq 1$ for all $j \notin T$. 
    Since $\hat{c}_i = c'_i + 1$ for all $i \in T'$ condition (3) of Corollary~\ref{cor:eps-close} holds. 
    Finally, if $c_i + 1 = c'_i$ for exactly one $i \in T'$ we simply set $\hat{c}=c$. 
    It clearly holds that $\hat{c}_i=c_i$ for all $i \in T \cap T'$, we have $\hat{c}_j=0$ for all $j \notin T$, and condition (1) of Corollary~\ref{cor:eps-close} holds between $\hat{c}$ and $c'$. 
    %
\end{proof}

Next, we analyze the error compared to the non-private MG sketch. 
We state the error in terms of the largest error among all elements of the sketch.
Recall that we implicitly say that the count is zero for any element not in the sketch.

\begin{lemma} \label{lem:PMG-error}
    Let $\tilde{T}, \tilde{c} \leftarrow \PMG(T, c)$ denote the output of Algorithm~\ref{alg:PMG} for any sketch $T, c$ with $\vert T \vert = k$. 
    Then with probability at least $1 - \beta$ we have 
    \begin{align*}
        \tilde{c}_x \in \Bigg[c_x - \frac{2\ln\big(\frac{k+1}{\beta}\big)}{\epsilon} - 1 - \frac{2\ln\big(3/\delta\big)}{\epsilon}, 
     c_x + \frac{2\ln\big(\frac{k+1}{\beta}\big)}{\epsilon} \Bigg] 
    \end{align*}
    for all $x \in T$ and $\tilde{c}_x = 0$ for all $x \notin T$. 
\end{lemma}

\begin{proof}
    The two sources of error are the noise samples and the thresholding step.
    We begin with a simple bound on the absolute value of the Laplace distribution.

    \[
        \Pr\bigg[\vert \Lap{1/\epsilon} \vert \geq \frac{\ln((k+1)/\beta)}{\epsilon}\bigg] = 
        2 \cdot \Pr\bigg[ \Lap{1/\epsilon} \leq -\frac{\ln((k+1)/\beta)}{\epsilon}\bigg] = \beta/(k+1) \, .
    \]
    Since $k + 1$ samples are drawn we have by a union bound that the absolute value of all samples is bounded by $\ln((k+1)/\beta)/\epsilon$ with probability at least $1 - \beta$.
    As such the absolute error from the Laplace samples is at most $2\ln((k+1)/\beta)/\epsilon$ for all $x \in T$ since two samples are added to each count.
    Removing noisy counts below the threshold potentially adds an additional error of at most $1 + \thres$.
    It is easy to see that $\tilde{c}_x = 0$ for all $x \notin T$ since the algorithm never outputs any such elements. 
\end{proof}

\begin{theorem} \label{thm:final_theorem}
    $\PMG(k, \x)$ satisfies $(\epsilon,\delta)$-differential privacy. 
    Let $f(x)$ denote the frequency of any element $x \in \Universe$ in $S$ and let $\hat{f}(x)$ denote the estimated frequency of $x$ from the output of $\PMG(k, \x)$.
    For any $x$ with $f(x)=0$ we have $\hat{f}(x)=0$ and with probability at least $1 - \beta$ we have for all $x \in \Universe$ 
    \[
    \hat{f}(x) \in \left[ f(x) - \frac{2\ln\left(\frac{k+1}{\beta}\right)}{\epsilon} - 1 - \frac{2 \ln (3/\delta)}{\epsilon} - \frac{\vert \x \vert}{k+1},
    f(x) + \frac{2\ln\left(\frac{k+1}{\beta}\right)}{\epsilon} \right] \, .
    \]
    
    Moreover, the algorithm outputs all $x$, such that $\hat{f}(x) > 0$ and there are at most $k$ such elements. 
    $\PMG(k,\x)$ uses $2k$ words of memory.
    The mean squared error for any fixed element $x \in \Universe$ is bounded by 
    $\E[(\hat{f}(x) - f(x))^2] \leq 3\left(1 + \frac{2 + 2 \ln(3/\delta)}{\epsilon} + \frac{|\x|}{k+1}\right)^2$.
    
\end{theorem}

\begin{proof}
    The space complexity is clearly as claimed, as we are storing at any time at most $k$ items and counters. 
    We focus on proving privacy and correctness.

    If $f(x) = 0$ we know that $x \notin T$ where $T$ is the keyset after running Algorithm~\ref{alg:MG}. 
    Since Algorithm~\ref{alg:PMG} outputs a subset of $T$ we have $\hat{f}(x) = 0$.
    The first part of the Theorem follows directly from Fact~\ref{fact:MG-error} and Lemmas \ref{lem:PMG-privacy} and~\ref{lem:PMG-error}. 

    We now bound the mean squared error. There are three sources of error. 
    Let $r_1$ be the error coming from the Laplace noise, $r_2$ from the thresholding, and $r_3$ the error made by the MG sketch. Then 
    \[
    \E[(\hat{f}(x) - f(x))^2] = \E[(r_1 + r_2 + r_3)^2] \leq 3(\E[r_1^2] + \E[r_2^2] + \E[r_3^2])
    \]
    by equivalence of norms (for any dimension $n$ vector $v$, $\|v\|_1 \leq \sqrt{n}\|v\|_2$).
    The errors $r_2,r_3$ are deterministically bounded $r_2 \leq 1+2\ln(3/\delta)/\epsilon$ and $r_3 \leq |\x|/(k+1)$. $\E[r_1^2]$ is the variance of the Laplace noise; we added two independent noises each with 
    scale $1/\epsilon$ and thus 
    variance $2/\epsilon^2$ for a total variance of $4/\epsilon^2$. This finishes the proof.
\end{proof}

\subsection{Privatizing standard versions of the Misra-Gries sketch}
\label{sec:MG-without-ordering}

The privacy of our mechanism as presented in Algorithm~\ref{alg:PMG} relies on our variant of the Misra-Gries algorithm.
Our sketch can contain elements with a count of zero. 
However, elements with a count of zero are removed in the standard version of the sketch.
As such, sketches for neighboring datasets can differ for up to $k$ keys if one sketch stores $k$ elements with a count of $1$ and the other sketch is empty.
It is easy to change Algorithm~\ref{alg:PMG} to handle these implementations. 
We simply increase the threshold to $1 + 2\ln\left(\frac{k+1}{2\delta}\right)/\epsilon$ since the probability of outputting any of the $k$ elements with a count of $1$ is bounded by $\delta$.

\subsection{Tips for practitioners}
\label{sec:practitioners}

Here we discuss some technical details to keep in mind when implementing our mechanism.

The output of the Misra-Gries algorithm is an associative array. 
In Algorithm~\ref{alg:PMG} we add appropriate noise such that the associative array can be released under differential privacy.
However, for some implementations of associative arrays such as hash tables the order in which keys are added affects the data structure.
Using such an implementation naively violates differential privacy, but it is easily solved either by outputting a random permutation of the key-value pairs or using a fixed order e.g. sorted by key.

We present our mechanism with noise sampled from the Laplace distribution. 
However, the distribution is defined for real numbers which cannot be represented on a finite computer. 
This is a known challenge and precision-based attacks still exist on popular implementations~\cite{DP-floating-point-attacks}.
Since the output of $\MG$ is discrete the distribution can be replaced by the Geometric mechanism~\cite{DP-geometric-mechanism} or one of the alternatives introduced in \cite{DP-finite-computers}.
Our mechanism would still satisfy differential privacy but it might be necessary to change the threshold in Algorithm~\ref{alg:PMG} slightly to ensure that Lemma~\ref{lem:nonsupport-delta} still holds.
Our proof of Lemma~\ref{lem:nonsupport-delta} works for the Geometric mechanism from~\cite{DP-geometric-mechanism} when increasing the threshold to $1 + 2\ceil{\ln(6e^\varepsilon/((e^\varepsilon + 1)\delta))/\varepsilon}$.

Lastly, it is worth noting that the analysis for Lemma~\ref{lem:nonsupport-delta} is not tight. 
We bound the probability of outputting a small key by bounding the value of all relevant samples by $\ln(3/\delta)/\epsilon$ which is sufficient to guarantee that the sum of any two samples does not exceed $2\ln(3/\delta)/\epsilon$. 
This simplifies the proof and presentation significantly however one sample could exceed $\ln(3/\delta)/\epsilon$ without any pair of samples exceeding $2\ln(3/\delta)/\epsilon$. 
A tighter analysis would improve the constant slightly which might matter for practical applications.

\section{Pure Differential Privacy}
\label{sec:pure}

In this section, we discuss how to achieve $\epsilon$-differential privacy.
We cannot use our approach from Section~\ref{sec:alg} where we add the same noise to all keys because the set of stored keys can differ between sketches for neighboring datasets.
Instead, we achieve privacy by adding noise to all elements of $\Universe$ scaled to the $\ell_1$-sensitivity.
Chan et al.~\cite{DP-continual-heavy-hitters} showed that the sensitivity of Misra-Gries sketches scales with the number of counters. 
We show that a simple post-processing step reduces the sensitivity of the sketch to $2$ and the worst-case error of the sketch is still $n/(k+1)$ where $n = \vert \x \vert$. 
This allows us to achieve an error of $n/(k+1) + O(\log (d)/\epsilon)$.

The $\ell_1$-sensitivity scales with the size of the MG sketch since the counts can differ by 1 for all $k$ elements between neighboring datasets. 
This happens for any pair of neighboring streams where the decrement step of Algorithm~\ref{alg:MG} (Branch 2) is executed one more time for one stream compared to the other.
We get around this case by post-processing the sketch before adding noise.
We subtract a value from each counter in the sketch that depends on the number of times Branch 2 of Algorithm~\ref{alg:MG} was executed.
Although we increase the error of the MG sketch, the worst-case error guarantee remains the same, since at most $n/(k+1)$ has been subtracted from any count.
We present the pseudocode of our approach in Algorithm~\ref{alg:MG-pure-DP}.
Next, we show that the worst-case error guarantees of Algorihtm~\ref{alg:MG-pure-DP} match those of the MG sketch. 
We then bound the $\ell_1$-sensitivity.

{
\begin{figure}[h]
\centering
\begin{minipage}{.85\linewidth}

\begin{algorithm}[H]
\caption{Misra-Gries Sketch Sensitivity Reduction \label{alg:MG-pure-DP}}
\SetKwInOut{Input}{Input}

\Input{Positive integer $k$ and stream $\x \in \Xsingle$}

$T, c \leftarrow \text{MG}(k, \x)$ \\ 
$\gamma \leftarrow {\sum_{x \in T} c_x/(k+1)}$ \Comment{\makebox[6cm][l]{Compute offset value}} \\
$\hat{T} \leftarrow \emptyset$ \\
\ForEach{$x \in T$}
{
        \If(\Comment{\makebox[6cm][l]{Removes negative frequency estimates}}){$c_x > \gamma$} { \label{line:sens-reduction-condition}
          $\hat{T} \leftarrow \hat{T} \cup \{x\}$ \\
          $\hat{c}_x \leftarrow c_x - \gamma$ \Comment{\makebox[6cm][l]{Apply offset to all counts in the sketch}}
        }
}
\Return{$\hat{T}, \hat{c}$}
\end{algorithm}

\end{minipage}
\end{figure}
}

\begin{lemma}
    \label{lem:pure-dp-postprocessing-error}
    Let $\hat{f}(x)$ be the frequency estimates given by Algorithm~\ref{alg:MG-pure-DP} for size $k$. 
    Then for all $x \in \mathcal{U}$, it holds that $\hat{f}(x) \in [f(x) - \vert S\vert/(k+1), f(x)]$, where $f(x)$ is the true frequency of $x$ in $\x$.
\end{lemma}

\begin{proof}
    Let $\alpha$ denote the number of times Branch 2 of Algorithm~\ref{alg:MG} was executed.
    Then the frequency estimate of any $x \in \mathcal{U}$ before post-processing must be in $[f(x) - \alpha, f(x)]$, since only the decrement step of the MG sketch introduces any error.
    Note that we have $\alpha \leq \vert S \vert/(k+1)$ by Fact~\ref{fact:MG-error}.
    Notice that $\sum_{x \in T} c_x = \vert S \vert - \alpha \cdot (k + 1)$ since $k$ counters are decremented and $1$ element is ignored each time Branch 2 of Algorithm~\ref{alg:MG} is executed.
    As such, we have $\gamma = \vert S \vert/(k+1) - \alpha$ and $c_x - \gamma \in [f(x) - \vert S \vert/(k+1), f(x) - \vert S \vert/(k+1) + \alpha]$.
    The claim therefore holds for any $x \in \hat{T}$ since $\hat{c}_x = c_x - \gamma$.
    For any key that is removed by post-processing, that is $x \in T \setminus \hat{T}$, it holds that $c_x - \gamma \leq 0$.
    This implies that the true frequency $f(x)$ is at most $\vert S \vert/(k + 1)$ so the error guarantee still holds.     
    Since any $x \notin T$ is unaffected by post-processing, those frequency estimates still lie in the interval $[f(x) - \alpha, f(x)]$.
\end{proof}

\begin{lemma}
    \label{lem:pure-dp-postprocessing-sensitivity}
    Let $\hat{T},\hat{c}$ and $\hat{T}',\hat{c}'$ denote the output of Algorithm~\ref{alg:MG-pure-DP} for a pair of neighboring streams.
    Consider $\hat{c} \in \mathbb{R}^d$ where we define $\hat{c}_x = 0$ for all $x \notin \hat{T}$ and define $\hat{c}'$ similarly.
    Then we have $\|\hat{c} - \hat{c}'\|_1 < 2$. 
\end{lemma}

\begin{proof}
    Let $\x \sim \x'$ denote any pair of neighboring streams where $\x'$ is obtained by removing one element from $\x$.
    The proof is symmetric if $\x'$ is obtained by adding one element to $\x$.
    We show that the sensitivity is bounded by considering the effect of our proposed post-processing for both cases of Lemma~\ref{lem:nb-cases}.
    Notice that we can ignore any counter that has a value of zero in any sketch, because they are always removed by the condition on Line~\ref{line:sens-reduction-condition}.

    We first consider the case where $c_i = c'_i - 1$ for all $i \in T'$.
    Here we have $\sum_{x \in T} c_x = \sum_{x \in T'} c'_x - k$ which implies that $\gamma = \gamma' - k/(k+1) = \gamma' - 1 + 1/(k+1)$.
    As such, we have for any $i \in T'$ that $c_i - \gamma = c'_i - \gamma' - 1/(k+1)$.
    From this it follows that for any $i \in \hat{T}'$ we have $\hat{c}_i - \hat{c}'_i = -1/(k+1)$, and for any $j \notin \hat{T}'$ we have $\hat{c}_j = \hat{c}'_j = 0$. 
    We can thus bound the distance between the vectors by $\| \hat{c} - \hat{c}' \|_1 = \sum_{i \in [d]} \vert \hat{c}_i - \hat{c}_i' \vert = \sum_{x \in \hat{T}'} \vert \hat{c}_x - \hat{c}'_x \vert = \vert \hat{T}' \vert \cdot 1/(k+1) \leq k /(k+1) < 1$.

    Next, we consider the case where sketches differ in only one counter before post-processing. 
    That is, there exists an element $x \in \mathcal{U}$ such that $c_x = c'_x + 1$ and all other counters are equal.
    In this case we have $\gamma = \gamma' - 1/(k+1)$ 
    and $c_i - \gamma = c'_i - \gamma' + 1/(k+1)$ for any $i \neq x$.
    As such, we either have $\hat{c}_i - \hat{c}'_i = 1/(k+1)$ or $\hat{c}_i = \hat{c}'_i = 0$ for each $i \in \Universe \setminus \{x\}$.
    Furthermore, we have that $0 \leq \hat{c}_x - \hat{c}'_x \leq 1 + 1/(k+1)$.
    We conclude the proof by bounding the distance between the vectors by $\| \hat{c} - \hat{c}' \|_1 = \sum_{i \in [d]} \vert \hat{c}_i - \hat{c}_i' \vert = \sum_{i \in \hat{T}} \vert \hat{c}_i - \hat{c}'_i \vert = \vert \hat{T} \setminus \{x\} \vert \cdot 1/(k+1) + \vert \hat{c}_x - \hat{c}'_x \vert \leq 1 + k /(k+1) < 2$.
\end{proof}

We achieve $\epsilon$-differential privacy by adding noise to our new sketch. 
We essentially use the same technique as Chan et al.~\cite{DP-continual-heavy-hitters} but the noise no longer scale linearly in $k$ as the sensitivity of Algorithm~\ref{alg:MG-pure-DP} is bounded by $2$.
Specifically, we add noise sampled from $\Lap{2/\epsilon}$ independently to the count of each element from $\Universe$ and release the top-$k$ noisy counts.
A simple union bound shows that with probability at least $1-\beta$ the absolute value of all samples is bounded by $2\ln(d/\beta)/\epsilon$.
Note that it might be infeasible to actually sample noise for each element when $\Universe$ is large; we refer to previous work on how to implement this more efficiently~\cite{DP-continual-heavy-hitters, DP-pure-sparse-hist, DP-finite-computers}.

It is worth noting that the low sensitivity of the post-processed sketch can also be utilized under $(\epsilon, \delta)$-differential privacy. 
We can use an approach similar to \cite{DP-approx-sparse-hist}. 
They add noise to all non-zero counters and remove noisy counts below a threshold to hide small counters. 
Applying the standard approach for histograms would require a threshold with a small dependence on $k$ as neighboring sketches might disagree on all keys.  
However, \cite[Algorithm~9]{DP-ALP} extended the technique to real-valued vectors by probabilistically rounding elements with a value less than the $\ell_1$-sensitivity.
If we apply their technique directly we get a threshold of $4 + 2\ln(1/\delta)/\epsilon$.
As such, we can achieve error guarantees that match those from Theorem~\ref{thm:final_theorem} up to constant factors.
However, this approach has worse guarantees than Algorithm~\ref{alg:PMG} when comparing to the non-private Misra-Gries sketch.
By Lemma~\ref{lem:PMG-error} the error of Algorithm~\ref{alg:PMG} when compared to the non-private MG sketch is $O(\log(1/\delta)/\epsilon)$ with high probability (for sufficiently small $\delta$).
Here the error is $n/(k+1) + O(\log(1/\delta)/\epsilon)$ since we subtract up to $n/(k+1)$ from the counters before adding noise.

\section{Privatizing merged sketches}
\label{sec:merge}

In practice, it is often important that we may merge sketches. 
This is for example commonly used when we have a dataset distributed over many servers. 
Each dataset consists of multiple streams in this setting, and we want to compute an aggregated sketch over all streams.
We say that datasets are neighboring if we can obtain one from the other by removing a single element from one of the streams.
If the aggregator is untrusted we must add noise to each sketch before performing any merges.
This is the setting in~\cite{DP-continual-heavy-hitters} and we can run their merging algorithm.
However, since we add noise to each sketch the error scales with the number of sketches.
In particular, the error from the thresholding step of Algorithm~\ref{alg:PMG} scales linearly in the number of sketches for worst-case input.
In the rest of this section, we consider the setting where aggregators are trusted. 
We can achieve optimal error in this setting by applying the post-processing step from the previous section to each sketch before aggregating the counters of each element. 
The $\ell_1$-sensitivity of the aggregated sketch is still bounded by 2.
However, the aggregated sketch might have much more than $k$ counters. 
This approach increases the memory requirement of the aggregator linearly with the total number of sketches in the worst case. 
Next we consider a merging algorithm where aggregators never store more than $2k$ counters.

Agarwal, Cormode, Huang, Phillips, Wei, and Yi~\cite{mergeable-summaries} introduced the following simple merging algorithm in the non-private setting. 
Given two Misra-Gries sketches $T_1, c_1 \leftarrow \MG(k, \x^{(1)})$ and $T_2, c_2 \leftarrow \MG(k, \x^{(2)})$ they first compute the sum of all counters $c_1 + c_2$. 
There are up to $2k$ counters at this point. 
They subtract the value of the $k + 1$'th largest counter from all elements.
Finally, any non-positive counters are removed leaving at most $k$ counters.
They show that merged sketches have the same worst-case guarantee as non-merged Misra-Gries sketches. 
That is, if we compute a Misra-Gries sketch for each stream $(S^{(1)},\ldots,S^{(\numstreams)})$ and merge them into a single sketch, the frequency estimate of all elements is at most $N/(k+1)$ less than the true frequency. 
Here $N$ is the total length of all streams.
This holds for any order of merging and the streams do not need to have the same length.


Unfortunately, the structure between neighboring sketches where either a single counter or exactly $k$ counters differ by $1$ breaks down when merging.
Therefore we cannot run Algorithm~\ref{alg:PMG} on the merged sketch.
However, as we show below, the global sensitivity of merged sketches is independent of the number of merges.
The sensitivity only depends on the number of counters.
We first show a property for a single merge operation; this will allow us to bound the sensitivity for any number of merges.
Note that unlike in Section~\ref{sec:alg}, we do not limit the number of keys that differ between sketches and we do not store keys with a count of zero. 

\begin{lemma} \label{lem:merge-step-sensitivity}
    Let $T_1, c_1$, $T'_1, c'_1$ and $T_2, c_2$ denote Misra-Gries sketches of size $k$ and denote the sketches merged with the algorithm above as $\hat{T}, \hat{c} \leftarrow \merge (T_1, c_1, T_2, c_2)$ and $\hat{T}', \hat{c}' \leftarrow \merge (T'_1, c'_1, T_2, c_2)$. 
    If $T'_1 \subseteq T_1$ and $c_{1i} - c'_{1i} \in \{0, 1\}$ for all $i \in \Universe$ then at least one of the following holds (1) $\hat{T}' \subseteq \hat{T}$ and $\hat{c}_{i} - \hat{c}'_{i} \in \{0, 1\}$ for all $i \in \Universe$ or (2) $\hat{T} \subseteq \hat{T}'$ and $\hat{c}'_{i} - \hat{c}_{i} \in \{0, 1\}$ for all $i \in \Universe$.
\end{lemma}

\begin{proof}
    Let $\bar{c}$ and $\bar{c}'$ denote the merged counters before subtracting and removing values.
    Then clearly $\bar{c}_i - \bar{c}'_i \in \{0, 1\}$ for all $i \in \Universe$.
    Therefore we have that $\bar{c}_{k+1} - \bar{c}'_{k+1} \in \{0 , 1\}$ where $\bar{c}_{k+1}$ denotes the value of the $k + 1$'th largest counter in $\bar{c}$.
    Note that it does not matter if the $k+1$'th largest counter corresponds to a different element.
    If $\bar{c}_{k+1} = \bar{c}'_{k+1}$ we subtract the same value from each sketch and we have $\hat{c}_i - \hat{c}'_i \in \{0, 1\}$ for all $i \in \Universe$ and also $\hat{T}' \subseteq \hat{T}$.
    If $\bar{c}_{k+1} - \bar{c}'_{k+1} = 1$ we subtract one more from each count in $\hat{c}$ and we have $\hat{c}'_i - \hat{c}_i \in \{0, 1\}$ for all $i \in \Universe$ and $\hat{T} \subseteq \hat{T}'$. 
\end{proof}

\begin{corollary} \label{corol:merged-sensitivity}
    Let $(S^{(1)}, \ldots, S^{(\numstreams)})$ and $(S'^{(1)}, \ldots, S'^{(\numstreams)})$ denote two sets of streams such that $S^{(i)} \sim S'^{(i)}$ for one $i \in [\numstreams]$ and $S^{(j)} = S'^{(j)}$ for any $j \neq i$.
    Let $T, c$ and $T', c'$ be the result of merging Misra-Gries sketches computed on both sets of streams in any fixed merging order.
    Then $c$ and $c'$ differ by 1 for at most $k$ elements and agree on all other counts.
\end{corollary}

\begin{proof}
    The condition of Lemma~\ref{lem:merge-step-sensitivity} clearly holds for sketches of a pair of neighboring datasets by Lemma~\ref{lem:nb-cases} if we remove all elements with a count of zero. 
    The claim holds by induction after each merging operation by Lemma~\ref{lem:merge-step-sensitivity}.
\end{proof}

Since the $\ell_1$-sensitivity is $k$ we can use the algorithm in~\cite{DP-continual-heavy-hitters} that adds noise with magnitude $k/\epsilon$ to all elements in $\Universe$ and keeps the top-$k$ noisy counts.
If we only add noise to non-zero counts we can hide that up to $k$ keys can change between neighboring inputs with a threshold. 
The two approaches have expected maximum error compared to the non-private merged sketch of $O(k \log (d) / \epsilon)$ and $O(k \log (k/\delta) / \epsilon)$, respectively. 
For $(\varepsilon, \delta)$-differential privacy we can also utilize that the $\ell_2$-sensitivity is $\sqrt{k}$ because counters only differ by $1$.
This allows us to add noise that scales slower with $k$ using the Gaussian Sparse Histogram Mechanism~\cite{GaussianSparseHistogramMechanism} which we discuss further in the next section.

\section{User-level Differential Privacy} 
\label{sec:user-level}

So far we considered the setting where a user has exactly one element. 
In this section, we consider the more general setting where a user contributes up to $m$ distinct elements.
Specifically, this means that the input is a stream $\x = (\x_1, \x_2, \dots, \x_{n-1}, \x_{n})$, where for all $i \in [n]$ we have $\x_i \subseteq \Universe$ and $\vert \x_i \vert \leq m$. 
Our goal is to estimate the frequency of each element, that is $f(x) = \sum_{i \in [n]} \mathbbm{1}[x \in \x_i]$. 
We denote the total length of the stream as $N = \sum_{i \in [n]} \vert \x_i \vert$.
If we want to compute a MG sketch in this setting we have to flatten the input.
We denote by $\hat{\x}$ a stream created by processing the sets in $\x$ one at a time by iterating over each element in the set in some fixed order (e.g. ascending order).

We first show that Algorithm~\ref{alg:PMG} satisfies differential privacy in this setting.
The magnitude of noise and the threshold increase as a function of $m$.
We then introduce a new sketch with similar error guarantees to the MG sketch that can be released with less noise than Algorithm~\ref{alg:PMG} for many parameters in this setting.
Note that throughout this section we assume that $k > m$. 
This is a fair assumption because the error guarantees of the MG sketch are meaningless when $m > k$ because $f(x) \leq n$ by definition and $N/(k + 1) = (n \cdot m)/(k + 1) \geq n$ if $\vert S_i \vert = m > k$ for all $i \in [n]$.

We start by restating the group privacy property of differential privacy in our context.


\begin{lemma}[Group Privacy (Following \cite{DworkRothBook})]
    \label{lem:group-privacy}
    Let $\mathcal{M}$ be a mechanism satisfying $(\varepsilon, \delta)$-differential privacy where one stream of a neighboring pair of streams is obtained by removing $1$ element from the other.
    Then $\mathcal{M}$ satisfies $(m\varepsilon, m e^{m\varepsilon} \delta)$-differential privacy for neighboring streams where one stream is obtained by removing at most $m$ elements from (or adding at most $m$ elements to) the other stream.
\end{lemma}

This allows us to upper bound the privacy parameters when streams differ by multiple elements.

\begin{lemma}
    \label{lem:PMG-multiple-items} 
    Let $\hat{\x}$ be the flattened version of the stream $\x$.
    Then releasing $\PMG(k, \hat{\x})$ with parameters $\varepsilon = \varepsilon'/m$ and $\delta = \delta'/(m e^{\varepsilon'})$ satisfies $(\varepsilon', \delta')$-differential privacy.
\end{lemma}

\begin{proof}
    $\PMG(k, \hat{\x})$ satisfies $(\varepsilon, \delta)$-differential privacy for streams differing in a single element by Lemma~\ref{lem:PMG-privacy}.
    For any neighboring streams $\x \sim \x'$ we have that $\hat{\x}$ can be created by removing at most $m$ elements from $\hat{\x}'$, or vice versa.
    It follows from Lemma~\ref{lem:group-privacy} that $\PMG(k, \hat{\x})$ satisfies $(\varepsilon', \delta')$-differential privacy as $m\varepsilon=\varepsilon'$ and $me^{m\varepsilon}\delta=me^{\varepsilon'}\delta=\delta'$ by definition.
\end{proof}

Since the proof above does not depend on the fact that the elements of a user are distinct or appear in consecutive order the algorithm can be used in an even more general setting.

\begin{corollary}
    \label{cor:PMG-multiple-items-general-setting}
    Algorithm~\ref{alg:PMG} with the parameters of Lemma~\ref{lem:PMG-multiple-items} satisfies $(\varepsilon', \delta')$-differential privacy under any definition of neighboring streams where one stream is obtained from the other by adding or removing up to $m$ (possibly duplicate) elements. 
\end{corollary}

We can use a similar argument to achieve pure differential privacy when streams differ by multiple elements using our technique from Section~\ref{sec:pure}.

\begin{lemma}
    \label{lem:pure-dp-user-level}
    Let $\x \in \Xsingle$
    be a stream of elements where a neighboring stream $\x'$ is obtained from $\x$ by adding or removing up to $m$ (possibly duplicate) elements.
    If we compute a MG sketch of $\x$ followed by the post-processing step from Section~\ref{sec:pure}, we can release the sketch under $\varepsilon$-DP using the technique of \cite{DP-continual-heavy-hitters} when noise is sampled from $\Lap{2m/\varepsilon}$.
\end{lemma}

\begin{proof}
    It follows from Lemma~\ref{lem:group-privacy} since the mechanism satisfies $(\varepsilon / m)$-differential privacy for streams differing by $1$ element because the $\ell_1$-sensitivity is bounded by $2$ as discussed in Section~\ref{sec:pure}.
\end{proof}

We cannot hope to avoid scaling the noise linearly with $m$
in the generalized setting of Corollary~\ref{cor:PMG-multiple-items-general-setting} and Lemma~\ref{lem:pure-dp-user-level} where users can contribute duplicate elements to the stream. 
The reason is that any frequency can differ by $m$ between neighboring streams if a user contributes only identical elements.
We instead focus on the setting where elements are distinct.
In this setting, additional improvements are possible under $(\varepsilon,\delta)$-DP because the $\ell_2$ sensitivity is smaller, as we discuss below.

Algorithm~\ref{alg:PMG} and the mechanism by \cite{DP-continual-heavy-hitters} relies on noise from the Laplace distribution. 
This choice works well for the settings of Sections~\ref{sec:alg} and~\ref{sec:pure}. 
However, when users have multiple distinct items we can instead use Gaussian noise under $(\varepsilon, \delta)$-differential privacy. 
The magnitude of Gaussian noise required for differential privacy scales with the $\ell_2$-sensitivity rather than the $\ell_1$-sensitivity which is typically used for Laplace noise.
If we consider the task of releasing $f(x)$ for all $x \in \Universe$ without memory constraints it is clear that a user changes at most $m$ distinct counts by $1$ each.
The $\ell_1$-sensitivity of $f$ is therefore $m$, but the $\ell_2$-sensitivity is only $\sqrt{m}$. 


We next discuss a general mechanism for releasing sketches using Gaussian noise.
The mechanism is similar to the technique with Laplace noise by~\cite{DP-approx-sparse-hist}.
The idea is to add noise from a Gaussian distribution to all non-zeroes counters and then remove small noisy counts.
This is known as the Gaussian Sparse Histogram Mechanism. 
Karrer, Kifer, Wilkins, and Zhang~\cite{GaussianSparseHistogramMechanism} gave an exact analysis of the parameters required for the mechanism to satisfy $(\varepsilon, \delta)$-DP.
They consider a more general version of the mechanism, and we restate their result for our setting in Theorem~\ref{thm:sparse-gaussian}.

\begin{theorem}[{Gaussian Sparse Histogram Mechanism (GSHM) - Following~\cite[Theorem~5.4]{GaussianSparseHistogramMechanism}}]
    \label{thm:sparse-gaussian}
    Let $T, c$ and $T', c'$ be frequency sketches for streams $\x$ and $\x'$, such that for any neighboring pair $\x \sim \x'$, the counters $c$ and $c'$ differ for at most $l$ counts and agree on the count for all other elements.
    The differing counts are either (a) all $1$ higher in $c$ than in $c'$ or (b) all $1$ lower in $c$ than in $c'$.
    Then the Gaussian Sparse Histogram Mechanism, denoted $\gaussthreshold(T, c, l, \sigma, \tau)$, adds noise from $\mathcal{N}(0, \sigma^2)$ independently to each non-zero count in $c$ and removes all noisy counts below the threshold $1 + \tau$.
    The mechanism satisfies $(\varepsilon, \delta)$-differential privacy if and only if the following inequality holds
    \begin{align*}
    \delta \geq 
    \max \bigg[ &1 - \Phi \left( \frac{\tau}{\sigma} \right)^l, \\
    &\max_{j \in [l]} 1 - \Phi \left( \frac{\tau}{\sigma} \right)^{l - j} + \Phi \left( \frac{\tau}{\sigma} \right)^{l - j} \left[ \Phi\left(\frac{\sqrt{j}}{2\sigma} - \frac{(\varepsilon - \gamma)\sigma}{\sqrt{j}}\right) - e^{\varepsilon - \gamma} \Phi\left(-\frac{\sqrt{j}}{2\sigma} - \frac{(\varepsilon - \gamma)\sigma}{\sqrt{j}}\right) \right], \\
    &\max_{j \in [l]} \Phi\left(\frac{\sqrt{j}}{2\sigma} - \frac{(\varepsilon + \gamma)\sigma}{\sqrt{j}}\right) - e^{\varepsilon + \gamma} \Phi\left(-\frac{\sqrt{j}}{2\sigma} - \frac{(\varepsilon + \gamma)\sigma}{\sqrt{j}}\right)
    \bigg] ,
    \end{align*}
    where $\gamma = (l - j)\log \Phi \left(\frac{\tau}{\sigma}\right)$.
\end{theorem}

The inequality of Theorem~\ref{thm:sparse-gaussian} gives the tightest results, but it can be difficult to parse. 
We present simpler parameters for the Gaussian Sparse Histogram Mechanism below.
Note that our analysis is far from tight and any deployment of the GSHM should preferably set parameters using the exact analysis presented in Theorem~\ref{thm:sparse-gaussian}.
We present the version with worse parameters for presentation purposes only as it is easier to read off the asymptotic behavior of the mechanism.


\begin{lemma}
    \label{lem:gaussian-sparse-easy-version}
    Let $T, c$, $T', c'$, and $\gaussthreshold(T, c, l, \sigma, \tau)$ all be defined as in Theorem~\ref{thm:sparse-gaussian}. 
    $\gaussthreshold(T, c, l, \sigma, \tau)$ satisfies $(\varepsilon, \delta)$-differential privacy when $\varepsilon < 1$ for $\sigma = \sqrt{l 2 \ln(2.5/\delta)}/\varepsilon$ and $\tau = \sqrt{2\ln(2l/\delta)} \sigma$.
\end{lemma}

\begin{proof}
    We use a proof similar to that currently used for the Google Differential Privacy library~\cite{googlelibthreshold} where the budget for $\delta$ is split between the noise and the threshold.
    Let $E_0$ denote the event where only counters that are non-zero in both sketches are above the threshold.
    We use $\mathcal{M}(x)$ and $\mathcal{M}(x')$ to denote $\gaussthreshold(T, c, l, \sigma, \tau)$ and $\gaussthreshold(T', c', l, \sigma, \tau)$, respectively. We have that
    \begin{align*}
    \Pr[\mathcal{M}(x) \in Z] &= \Pr[\mathcal{M}(x) \in Z \vert E_0] \Pr[E_0] + \Pr[\mathcal{M}(x) \in Z \vert \neg E_0] \Pr[\neg E_0] \\
    &\leq \Pr[\mathcal{M}(x) \in Z \vert E_0] \Pr[E_0] + \delta/2 \\
    &\leq (e^\varepsilon \Pr[\mathcal{M}(x') \in Z \vert E_0] + \delta/2) \Pr[E_0] + \delta/2 \\
    &\leq e^\varepsilon \Pr[\mathcal{M}(x') \in Z \text{ and } E_0] + \delta/2 + \delta/2 \\
    &\leq e^\varepsilon \Pr[\mathcal{M}(x') \in Z] + \delta \enspace .
    \end{align*}
    The first inequality follows from a union bound and the fact that $\Pr[\mathcal{N}(0, \sigma^2) > t] \leq e^{-(t/\sigma)^2/2}$ for any positive $t$.
    The second inequality follows from \cite[Theorem A.1]{DworkRothBook} since $\Delta_2 = \sqrt{l}$.
\end{proof}

As seen above the noise and threshold for the Gaussian Sparse Histogram Mechanism scales with the $\ell_2$-sensitivity rather than the $\ell_1$-sensitivity. 
Adding or removing a user to or from the stream only changes the true frequency of any $x \in \mathcal{U}$ by $1$.
As such the $\ell_2$ distance between the frequencies for any pair of neighboring streams $\x \sim \x'$ is $\sqrt{\sum_{x \in \mathcal{U}} (f(x) - f(x'))^2} \leq \sqrt{m}$.
We might hope that the $\ell_2$-sensitivity of the MG sketch is low, or perhaps that we can use a technique similar to Algorithm~\ref{alg:PMG} with Gaussian noise.
Alternatively, we could remove small values using the post-processing from Section~\ref{sec:pure} if that guaranteed low $\ell_2$-sensitivity.
Unfortunately, these techniques cannot be used to reduce the dependence on $m$ because there exist neighboring streams where a single counter differs by $m$. 
That statement holds even if we perform the post-processing from Section~\ref{sec:pure}.

\begin{lemma}
    \label{lem:l2sens-MG-sketch}
    Let $T, c$ and $T', c'$ denote MG sketches of size $k$ for a pair of neighboring streams $\x \sim \x'$.
    There exist neighboring streams such that $c_x - c'_x = m$ for some element $x \in \Universe$.
\end{lemma}

\begin{proof}
    There are many such pairs of streams, here we give an example of a pair where all counters for elements other than $x$ are zero in both sketches.
    Let $\x_{k + 1}$ be the user that is removed from $\x$ to obtain $\x'$.
    We construct $(\x_1, \x_2, \dots, \x_{k-1}, \x_{k})$ such that it contains exactly $m$ copies of $k$ distinct elements that does not include $x$.
    This is always possible for any $m \leq k$ by cycling through the same order of $k$ elements and taking $m$ elements at a time. 
    E.g, for $m=2$ and items $\{1, 2, 3\}$ the stream could be $(\{1, 2\}, \{3, 1\}, \{2, 3\})$.
    If $\x_{k + 1}$ contains $m$ elements that are not in the sketch all counters are decremented and the sketch for $\x$ is empty after processing $\x_{k + 1}$. 
    Now, the rest of the stream contains only copies of the singleton $\{x\}$.
    After processing $\x_{k + 1 + m}$ the sketch for $\x$ has a single counter $c_x = m$ while the sketch for $\x'$ is empty.
    Furthermore, if we extend the streams further with more copies of $\{x\}$ we have $c_x = \vert \x \vert - k - 1$ and $c'_x = \vert \x \vert - k - 1 - m$.
\end{proof}

Since there exist neighboring streams for which the MG sketch differs by $m$ for a single counter, as shown above, we must add noise with magnitude scaled to $m$ to satisfy differential privacy. 
We also cannot hope to avoid these problematic inputs using some post-processing for small counters similar to Section~\ref{sec:pure} because the counter that differs by $m$ can be arbitrarily large for sufficiently long streams.
We instead use a different sketching algorithm to remove the linear dependency on $m$. 
In Algorithm~\ref{alg:user-level-sketch} we present a novel frequency sketch.
Our sketch is similar in spirit to the MG sketch as it is a generalization of the original MG sketch for sets.
The key difference is that we always increment all counters for the elements of a user and decrement all counters at most once per user instead of once per element.

\begin{figure}[H]
\centering
\begin{minipage}{.85\linewidth}

\begin{algorithm}[H]
\caption{Privacy-Aware Misra-Gries~($\usersketch$)\label{alg:user-level-sketch}}
\SetKwInOut{Input}{Input}

\Input{Positive integer $k$ and stream $\x = (\x_1, \x_2, \dots, \x_{n-1}, \x_n)$} 

$T \leftarrow \emptyset$ \\

\ForEach{$\x_i \in \x$}
{
  \ForEach(\Comment{\makebox[6cm][l]{Increment count of each element in $\x_i$}}){$x \in \x_i$}{ \label{line:alg3-first-loop-start}
     \uIf{$x \in T$}{ 
       $c_x \leftarrow c_x + 1$
     }
     \Else{
       $T \leftarrow T \cup \{x\}$ \\
       $c_x \leftarrow 1$ \label{line:alg3-first-loop-finish}
     }
  }
 \If(\Comment{\makebox[6cm][l]{The size can temporarily be up to $k + m$}}){$\vert T \vert > k$\label{line:user-level-condition}}{
   \ForEach(\Comment{\makebox[6cm][l]{Decrement all counters}}){$x \in T$}{
    $c_x \leftarrow c_x - 1$ \label{line:user-level-sketch-dec} \\
    \If(\Comment{\makebox[6cm][l]{Remove keys for zero counters}}){$c_x = 0$}{
      $T \leftarrow T \setminus \{x\}$
    }
   }
 }
}

\Return{$T, c$}
\end{algorithm}

\end{minipage}
\end{figure}

We start our analysis by showing that the error guarantees of Algorithm~\ref{alg:user-level-sketch} match the MG sketch.

\begin{lemma}
    \label{lem:user-level-sketch-error}
    Let $\hat{f}(x)$ be the frequency estimates given by Algorithm~\ref{alg:user-level-sketch} with size $k$. 
    Then for all $x \in \mathcal{U}$, it holds that $\hat{f}(x) \in [f(x) - \floor{N/(k+1)}, f(x)]$, where $f(x)$ is the true frequency of $x$ in $\x$ and $N = \sum_{i \in [n]} \vert S_i \vert$ is the total size of the stream.
\end{lemma}

\begin{proof}
    The proof is similar to the analysis of the MG sketch by \cite{MG-max-error}.
    Notice that similar to the MG sketch, we only introduce errors when decrementing counters.
    The error of any estimate $\hat{f}(x)$ is exactly the number of times the counter for $x$ was decremented on line~\ref{line:user-level-sketch-dec}.
    As such, we can bound the maximum error of all counters by the number of times the condition on line~\ref{line:user-level-condition} evaluates to true.
    Notice that whenever the condition is true the sum of all counters is decreased by at least $k + 1$.
    Since the sum is incremented exactly $N$ times and counters are never negative this happens at most $\floor{N/(k + 1)}$ times.
\end{proof}

Next, we bound the sensitivity of Algorithm~\ref{alg:user-level-sketch}.

\begin{lemma}
    \label{lem:user-level-sketch-sensitivity}
    Let $\T,\cc \leftarrow \usersketch(k, \x)$ and $\T',\cc' \leftarrow \usersketch(k, \x')$ be the outputs of Algorithm~\ref{alg:user-level-sketch} on a pair of neighboring streams $S \sim S'$.
    Then it holds that either, (1) $\T' \subseteq \T$ and $c_i - c'_i =\{0, 1\}$ for all $i \in \T$ or (2) $\T \subseteq \T'$ and $c'_i - c_i =\{0, 1\}$ for all $i \in \T'$.
\end{lemma}

\begin{proof}
    We first show that the condition holds after processing the user that is only present in one of the streams.
    Without loss of generality consider the case where $\x'$ is obtained by removing $\x_i$ from $\x$.
    Since the two cases of the lemma are symmetric the proof is the same if $\x$ is obtained by removing some user $\x'_i$ from $\x'$.
    Since the algorithm is deterministic the sketches for the two streams are identical after processing $\x_{i - 1}$.
    The sketch for $\x$ then processes $\x_i$.
    After running the loop on Lines~\ref{line:alg3-first-loop-start}-\ref{line:alg3-first-loop-finish} the count for each element in $\x_i$ is $1$ higher compared to the sketch for $\x'$.
    If the condition on Line~\ref{line:user-level-condition} evaluates to false we finished processing $\x_i$ and the state corresponds to case $(1)$ of the lemma.
    If the condition is true we decrement all counters by $1$ and remove elements with a count of $0$.
    The state then corresponds to case $(2)$ of the lemma.

    We now show that if one of the conditions holds before processing some user $\x_j$ one of the conditions holds also after processing $\x_j$.
    This proves the lemma by induction.
    Assume that case $(1)$ of the lemma holds before processing $\x_j$.
    The proof is identical for case $(2)$ due to symmetry.
    In each iteration of the loop on lines~\ref{line:alg3-first-loop-start}-\ref{line:alg3-first-loop-finish} we increment the same counter in both sketches which does not affect the condition of case $(1)$.
    As such, we only have to consider the effect of decrementing counters.
    If we either do not decrement the counters in any sketch or we decrement the counters in both sketches the state after processing $\x_j$ still corresponds to case $(1)$.
    However, if $\vert T \vert > k \geq \vert T' \vert$ we only decrement the counters in $c$.
    The state then corresponds to case $(2)$.
    Note that we never decrement the counters in $c'$ without decrementing the counters in $c$ since we have that $\vert T \vert \geq \vert T' \vert$.
\end{proof}

As mentioned in Section~\ref{sec:merge} it is often important in practice that we can merge sketches. 
A nice property of Algorithm~\ref{alg:user-level-sketch} is that the sensitivity has the structure we used to bound the sensitivity of merged sketches.
The $\usersketch$ sketch can be seen as a special case of the merging algorithm discussed in Section~\ref{sec:merge}.
The algorithm is equivalent to computing a MG sketch for each user $T_i, c_i \leftarrow \MG(k, \x_i)$ and merging each $T_i, c_i$ with $T, c$ to obtain a sketch for the full stream.
As such, merging sketches from Algorithm~\ref{alg:user-level-sketch} does not increase sensitivity.

\begin{corollary}
    \label{cor:merged-user-level}
    Let $(S^{(1)}, \ldots, S^{(\numstreams)})$ and $(S'^{(1)}, \ldots, S'^{(\numstreams)})$ denote sets of streams such that $S^{(i)} \sim S'^{(i)}$ for one $i \in [\numstreams]$ and $S^{(j)} = S'^{(j)}$ for any $j \neq i$.
    Let $T, c$ and $T', c'$ be the result of merging $\usersketch$ sketches computed on both sets of streams in any fixed order using the merging algorithm by \cite{mergeable-summaries} as discussed in Section~\ref{sec:merge}.
    Then it holds that either, (1) $\T' \subseteq \T$ and $c_i - c'_i =\{0, 1\}$ for all $i \in \T$ or (2) $\T \subseteq \T'$ and $c'_i - c_i =\{0, 1\}$ for all $i \in \T'$.
\end{corollary}

\begin{proof}
    It holds by induction and Lemma~\ref{lem:merge-step-sensitivity}.
    The condition required by Lemma~\ref{lem:merge-step-sensitivity} holds for sketches of any pair of neighboring streams by Lemma~\ref{lem:user-level-sketch-sensitivity}.
\end{proof}

\begin{lemma}
    \label{lem:user-level-merge-error}
    Let $T, c$ be the result of merging $\usersketch$ sketches of size $k$ computed on a set of streams using the merging algorithm from Section~\ref{sec:merge}.
    Let $f(x)$ be the true frequency of $x$ across all streams and let $M$ be the total number of elements across all streams. Then for all $x \in \mathcal{U}$ it holds that
    \[
        c_x \in \left[ f(x) - \frac{M}{k + 1}, f(x) \right] \,.
    \]
\end{lemma}

\begin{proof}
    Agarwal et al.~\cite{mergeable-summaries} show that this holds for the MG sketch in Lemma~1 and Theorem~1 of their paper.
    They used the properties of the MG sketch that error is only introduced when decrementing counters, and counters for $k + 1$ distinct elements are decremented each time.
    Since at least $k + 1$ counters for distinct elements are decremented in Algorithm~\ref{alg:user-level-sketch} the lemma follows from their proof. 
\end{proof}


Finally, we are ready to summarize our results for this section in the following theorem.

\begin{theorem}
    \label{thm:user-level}
    Let $(S^{(1)}, \ldots, S^{(\numstreams)})$ and $(S'^{(1)}, \ldots, S'^{(\numstreams)})$ denote sets of streams such that $S^{(i)} \sim S'^{(i)}$ for one $i \in [\numstreams]$ and $S^{(j)} = S'^{(j)}$ for any $j \neq i$.
    Let $T, c$ and $T', c'$ be the result of merging $\usersketch$ sketches of size $k$ computed on both sets of streams in any fixed order using the merging algorithm discussed in Section~\ref{sec:merge}.
    Then releasing the output of $\gaussthreshold(T, c, k, \sigma, \tau)$ where $\sigma$ and $\tau$ are chosen according to Theorem~\ref{thm:sparse-gaussian} satisfies $(\varepsilon, \delta)$-differential privacy. 
    For any $\varepsilon < 1$ we have that $\tau = O( \sqrt{k}\ln(k/\delta)/\varepsilon )$.
    Let $\hat{f}(x)$ be the estimate of the true frequency $f(x)$ across all streams and let $M$ denote the total number of elements across all streams.
    Then we have for all $x \in \mathcal{U}$ with probability at least $1 - 2\delta$ that 
    \[
    \hat{f}(x) \in \left[ f(x) - \frac{M}{k+1} - 2\tau - 1,
    f(x) + \tau \right] \, .
    \]
\end{theorem}

\begin{proof}
    The privacy guarantees follow directly from Corollary~\ref{cor:merged-user-level} and Theorem~\ref{thm:sparse-gaussian}.
    The three sources of error are the estimation error of the non-private sketch, the noise added to each counter, and the error from removing values below the threshold.
    The error of the merged sketch is at most $\frac{M}{k + 1}$ by Lemma~\ref{lem:user-level-merge-error}.
    From the first part of the condition of Theorem~\ref{thm:sparse-gaussian} we have that $1 - \Phi(\frac{\tau}{\sigma})^k \leq \delta$.
    That is, the $k$ samples of Gaussian noise are all bounded by $\tau$ with probability at least $1 - \delta$.
    Since the noise is symmetric we can use a union bound to bound the absolute value of the noise by $\tau$ with probability at least $1 - 2\delta$.
    Removing noisy counters below the thresholding adds error at most $1 + \tau$.
    The asymptotic behavior of $\tau$ for $\varepsilon < 1$ follows from the looser bound in Lemma~\ref{lem:gaussian-sparse-easy-version}.
\end{proof}

\section{Open Problem}
\label{sec:open-problems}

In Section~\ref{sec:user-level} we introduced a frequency sketch for a stream of users each with up to $m$ distinct elements.
Our sketch has error guarantees similar to the MG sketch and the $\ell_2$-sensitivity of the sketch with size $k$ is $\sqrt{k}$.
The main open problem that we leave is if there exists a sketch with similar error guarantees and $\ell_2$-sensitivity of $\sqrt{m}$ or $O\left(\sqrt{m}\right)$.
This would allow us to add noise to the sketch with magnitude matching the non-streaming setting.
More specifically, we ideally want a sketch with all the properties listed below for the setting of Section~\ref{sec:user-level}.

\textsc{Open Problem.} \textit{
Is it possible to design a sketch with the following properties (1) the sketch uses $O(k)$ space (2) the sketch (non-privately) returns a frequency oracle with error at most $N/(k+1)$ for any stream (3) there exists an $(\varepsilon, \delta)$-DP mechanism that releases the sketch with error $O(\sqrt{m}\log(1/\delta)/\varepsilon)$?
}

The $\ell_2$-sensitivity of our sketch exceeds $O\left(\sqrt{m}\right)$ because the number of elements that are decremented when the sketch is full can differ between neighboring sketches.
Any counter-based sketch with low sensitivity likely must ensure that the number of decremented counters remains stable.
During this project we experimented with variants that always decremented a fixed number of elements when full.
Unfortunately, those attempts yielded higher sensitivity that the PAMG as the sketches did not have the property that all counts differ by at most $1$ between neighboring streams.

\section*{Acknowledgment}
\label{sec:acknowledgement}
This work was supported by the VILLUM Foundation grant 16582.
We thank Rasmus Pagh for suggesting this problem and for helpful discussions. We thank Martin Aumüller for his valuable feedback.
Finally, we thank anonymous reviewers for suggestions that improved presentation.

\bibliographystyle{abbrv} 
\bibliography{literature}  
%
\end{document}